\documentclass[11pt]{article}
\usepackage[margin=1in]{geometry}

\usepackage[utf8]{inputenc}
\usepackage[numbers,sort]{natbib}
\usepackage{amsmath, amssymb, amsthm, thmtools, amsfonts, bm, bbm, thm-restate}
\usepackage{authblk}
\usepackage{algorithm}
\usepackage{algorithmicx}
\usepackage[noend]{algpseudocode}

\usepackage{todonotes}
\usepackage{multirow}
\usepackage{comment}
\usepackage{dsfont}

\usepackage{xcolor}
\definecolor{BrickRed}{rgb}{0.8,0.25,0.33}
\usepackage[colorlinks,citecolor=blue,linkcolor=BrickRed]{hyperref}
\usepackage{tikz} 
\usetikzlibrary{calc}
\usetikzlibrary{automata, positioning}
\usepackage[framemethod=tikz]{mdframed}
\usepackage{cleveref}
\usepackage{nicefrac}

\theoremstyle{plain}
\newtheorem{thm}{Theorem}[section]
\newtheorem{cor}[thm]{Corollary}

\newtheorem{fact}[thm]{Fact}
\newtheorem{lem}[thm]{Lemma}
\newtheorem{cla}[thm]{Claim}

\newtheorem{obs}[thm]{Observation}

\crefname{thm}{Theorem}{theorems}
\crefname{cla}{Claim}{claims}
\crefname{lem}{Lemma}{lemmas}
\crefname{fact}{Fact}{facts}

\newcommand{\E}{\mathbb{E}}
\newcommand{\R}{\mathbb{R}}
\newcommand{\bbP}{\mathbb{P}}
\newcommand{\Z}{\mathbb{Z}}
\newcommand{\eps}{\varepsilon}
\newcommand{\calB}{\mathcal{B}}
\newcommand{\calD}{\mathcal{D}}
\newcommand{\calG}{\mathcal{G}}
\newcommand{\calI}{\mathcal{I}}

\newcommand{\calP}{\mathcal{P}}
\newcommand{\calY}{\mathcal{Y}}

\newcommand{\optoff}{\mathrm{OPT}_\mathrm{off}}
\newcommand{\opton}{\mathrm{OPT}_\mathrm{on}}

\newenvironment{wrapper}[1]
{
	\smallskip
	\begin{center}
		\begin{minipage}{\linewidth}
			\begin{mdframed}[hidealllines=true, backgroundcolor=gray!20, leftmargin=0cm,innerleftmargin=0.375cm,innerrightmargin=0.375cm,innertopmargin=0.375cm,innerbottommargin=0.375cm,roundcorner=10pt]
				#1}
			{\end{mdframed}
		\end{minipage}
	\end{center}
	\smallskip
}

\makeatletter
\renewcommand\AB@affilsepx{\quad \protect\Affilfont}
\makeatother

\author[1]{Kristen Kessel}
\author[1]{Amin Saberi}
\author[2]{Ali Shameli}
\author[1]{David Wajc}
\affil[1]{Stanford University}
\affil[2]{MIT}

\title{The Stationary Prophet Inequality Problem}
\date{\vspace{-1.0cm}}

\begin{document}

\maketitle

\begin{abstract}
We study a continuous and infinite time horizon counterpart to the classic prophet inequality, which we term the stationary prophet inequality problem. Here, copies of a good arrive and perish according to Poisson point processes. Buyers arrive similarly and make take-it-or-leave-it offers for unsold items. The objective is to maximize the (infinite) time average revenue of the seller.

\smallskip

Our main results are pricing-based policies which (i) achieve a $1/2$-approximation of the optimal offline policy, which is best possible, and (ii) achieve a better than  $(1-1/e)$-approximation of the optimal online policy. Result (i) improves upon bounds implied by recent work of Collina et al. (WINE'20), and is the first optimal prophet inequality for a stationary problem. Result (ii) improves  upon a  $1-1/e$ bound implied by recent work of Aouad and Sarita{\c{c}} (EC'20), and shows that this prevalent bound in online algorithms is not optimal for this problem. 

%\smallskip 
%We further present improved bounds for multi-good generalizations. Perhaps surprisingly, all our bounds hold for inventory-constrained sellers who can only keep a limited stock, and therefore lose out on revenue both due to perishing of items, as well as discarding of overflow items.
%\todo{Finally, since our policies are pricing-based, they naturally imply truthful and approximately-maximal mechanisms for social welfare and revenue maximization.}
\end{abstract}

\pagenumbering{gobble}
\newpage 
\pagenumbering{arabic}

\section{Introduction}\label{sec:intro}

A ubiquitous challenge in market economics is decision making under uncertainty, addressed by the area of online algorithms. 
Should a firm sell an item to a buyer now, or reject their bid in favor of possibly higher future bids (at the risk of no such higher future bids arriving)?
Such dynamics were studied by probabilists in the area of optimal stopping theory as early as the 60s and 70s \cite{dynkin1963optimum,krengel1978,krengel1977semiamarts}, 
and have regained renewed interest in recent years in the online algorithms community, in large part due to their relevance to mechanism design.

\smallskip

A classic problem in the area is the single-item prophet inequality problem. Here, a buyer wishes to sell a single item, and buyers arrive in some order, with buyer $i$ making a take-it-or-leave it bid $v_i$ drawn from a (known) distribution $\calD_i$. 
The classic result of \citet{krengel1977semiamarts,krengel1978} asserts that there exists a $\nicefrac{1}{2}$-competitive algorithm, i.e., an algorithm whose expected value is at least $\nicefrac{1}{2}$ of the value $\E[\max_i v_i]$ obtained by a ``prophet'' who knows the future. This is optimal---no online algorithm has higher competitive ratio.
Shortly after, \citet{samuel-cahn1984} presented a $\nicefrac{1}{2}$-competitive \emph{posted-price} policy (i.e., selling the item to the first buyer bidding above some fixed threshold), foreshadowing a long line of work on such pricing-based policies.

\smallskip 

The classic single-item prophet inequality problem has been generalized to selling more complicated combinatorical structures, including, e.g., multiple items \cite{hajiaghayi2007automated,alaei2014bayesian,chawla2020static}, knapsacks \cite{feldman2016online,dutting2020prophet}, matroids and their intersections \cite{kleinberg2019matroid,babaioff2018matroid,feldman2016online,dutting2020prophet,chawla2010multi}, matchings \cite{ezra2020online,gravin2019prophet,feldman2014combinatorial,feldman2016online,dutting2020prophet}, and arbitrary downward-closed families \cite{rubinstein2016beyond}.
Most of this work has focused on approximating the offline optimum algorithm, with recent work also studying the (in)approximability of the optimal online algorithm by poly-time algorithms \cite{anari2019nearly,papadimitriou2021online} and particularly by posted-price policies \cite{niazadeh2018prophet}.
Much of the interest in prophet inequality problems, and specifically pricing-based policies, has been fueled by their implication of truthful mechanisms which approximately maximize social welfare and revenue, first observed in \cite{chawla2010multi} (see the surveys \cite{lucier2017economic,correa2019recent,hartline2013bayesian}, and  \cite{correa2019pricing} for the ``opposite'' direction).

\smallskip 

Despite this rich line of work on prophet inequalities and their use in online markets, 
one salient feature of motivating markets is missing in these problems' formulations: the repeated nature of the dynamics of such markets. 
In such markets, companies care less about their immediate returns than their average long-term rewards. 
Over such long time horizons, companies produce additional goods, while items which are not sold fast enough may expire.
Neither the long-term objective nor the dynamic nature of goods to sell is captured by traditional prophet inequality problems.

\smallskip

We introduce a continuous-time, infinite time horizon counterpart to the classic prophet inequality problem, which we term the \emph{stationary prophet inequality} problem. 
Here, goods are produced over time, where items of each good arrive according to Poisson processes, and, if unsold, perish according to Poisson processes. 
Buyers with different valuations for the different goods similarly enter the market according to Poisson processes. 
When a buyer $b$ arrives, a policy determines immediately which item (if any) to sell to $b$.
The objective is to maximize the infinite time horizon average reward of the policy, compared to the optimal offline or online policies.
(See \Cref{sec:prelims} for precise problem formulation.)

\smallskip 

A closely related problem of dynamic weighted matching was recently introduced by \cite{aouad2020dynamic,collina2020dynamic}.
In their problem, the market is not bipartite, and agents arrive over time and can be matched at any time before their (sudden and unpredictable) departure from the market. (Our problem is the special case of theirs where buyers' departure rate is infinite.)
The authors of \cite{collina2020dynamic,aouad2020dynamic} present algorithms which give $\frac{1}{4}(1-\nicefrac{1}{e})$- and $\nicefrac{1}{8}$-approximations of the optimal online and offline policies, respectively. 
For our single-good problem, their approaches yield $(1-\nicefrac{1}{e})$- and $(1-\nicefrac{1}{(e-1)})\approx 0.418$-approximations of the optimal online and offline policies, respectively (see \Cref{appendix:priorwork}).
%\todo{not quite the correct reference for the $1-1/e$ bound?}

\smallskip 

We ask to what extent these bounds can be improved, in particular, using \emph{posted-price} policies.

\subsection{Our contributions}

Our main results concern the single-good stationary prophet inequality problem. 

\smallskip

Our first result is a pricing-based policy which achieves a competitive ratio of $\nicefrac{1}{2}$. That is, this policy achieves an approximation of $\nicefrac{1}{2}$ of the value gained by the optimal offline policy for this problem---a bound we show is optimal for any policy (pricing-based or otherwise).
\begin{wrapper}
\begin{restatable}{thm}{HalfCompetitiveSingleGoodPolicy}
\label{thm:1/2-competitive-single-good-policy}
There exists a $\nicefrac{1}{2}$-competitive posted-price policy for the single-good stationary prophet inequality problem.
No online policy has competitive ratio greater than $\nicefrac{1}{2}$.%\todo{box}
\end{restatable}
\end{wrapper}

\Cref{thm:1/2-competitive-single-good-policy} is the first optimally competitive policy for a stationary prophet inequality problem. 
Our policy follows the approach given by \cite{collina2020dynamic}, whose analysis implies a $\nicefrac{1}{3}$-competitive ratio by comparing to a natural LP benchmark (see \Cref{appendix:priorwork}).
Our first technical contribution is a more queuing-theoretic analysis, via which we show that their policy is in fact $1-\nicefrac{1}{(e-1)}\approx 0.418$-competitive. Unfortunately, this is the best bound achievable using their approach; we show that for their LP benchmark, the above $1-\nicefrac{1}{(e-1)}$ bound is tight (see \Cref{sec:limitations-prior-LPs}).
Our second contribution is a new constraint, relying on another fundamental result in queuing theory, namely that Poisson arrivals ``see'' time averages (PASTA, \cite{wolff1982poisson}).
These queuing theoretic and approximation algorithmic ideas combined yield our optimally-competitive policy.

\smallskip 

We next turn to the approximability of the optimal \emph{online} policy, where we might hope to achieve higher approximation guarantees. 
For the classic prophet inequality problem, \citet{niazadeh2018prophet} show that pricing-based policies yield no better approximation of the optimal online policy than they do of the optimal offline policy. 
%In particular, they show that an approximation ratio of $1-\nicefrac{1}{e}$ is optimal for the i.i.d.~setting versus both benchmarks.
For the stationary prophet inequality problem, the same is not true; while our inapproximability result of \Cref{thm:1/2-competitive-single-good-policy} implies that no competitive ratio beyond $\nicefrac{1}{2}$ is possible, an algorithm of \cite{aouad2020dynamic} yields a $1-\nicefrac{1}{e} \approx 0.632$ approximation of the optimal online policy.
We prove that this latter natural bound, prevalent in the online algorithms literature, is not optimal for our problem, and present a pricing-based policy which breaks this bound.

\begin{wrapper}
\begin{restatable}{thm}{ApproximateSingleGoodPolicy}
\label{thm:0.656-approximate-single-good-policy}
There exists a posted-price policy for the single-good stationary prophet inequality problem which is a $0.656$-approximation of the optimal online policy in expectation.
\end{restatable}
\end{wrapper}

Our analysis compared to the LP benchmark of \cite{aouad2020dynamic} gives a simple $(1-\nicefrac{1}{e})$-approximation of the optimal online policy, which is tight for their LP. Here, we show that our new PASTA constraint and analysis allow us to break this ubiquitous bound.

\medskip

\noindent\textbf{Mechanism design implications.} 
    By standard connections to mechanism design, 
    our pricing-based policies immediately imply truthful mechanisms which approximate the social-welfare and revenue maximizing offline and online mechanisms (see \Cref{appendix:mechanism-design}).

\smallskip 
Finally, using the same algorithmic and analytic ideas, together with additional stochastic dominance results, we extend our approach from the single-good to the multi-good problem. For this natural generalization, we present a $\nicefrac{15}{56} \approx 0.267$-competitive policy, improving on the $\nicefrac{1}{8}=0.125$-competitive policy of \cite{collina2020dynamic}.
%\todo{phrasing? discussion of their analysis in appendix?} 

\begin{restatable}{thm}{CompetitiveMultiGoodPolicy}
\label{thm:competitive-multi-good-policy}
There exists a $\nicefrac{15}{56}$-competitive policy for the multi-good problem.
\end{restatable}

\medskip 

\noindent\textbf{Bounded inventory.}
    Surprisingly, all of our policies retain their approximation guarantees even when sellers have small inventory sizes, and must discard items of goods when more than some (small) number of items of said good are already available. In contrast, in \Cref{cla:upper-bound-approx-ratio-inventory-C} we show that sufficiently small inventory size does, however, limit achievable approximation (of any policy).

\section{Preliminaries}
\label{sec:prelims}

\textbf{Problem statement.}
In the stationary prophet inequality problem, a seller wishes to sell items of $n$ types of goods $\calG$, while (approximately) maximizing the seller's average gain over an infinite time horizon.
Items of good $i \in \calG$ are homogeneous and are supplied according to a Poisson process with rate $\lambda_i \in (0,\infty)$ and perish at an exponential rate $\mu_i \in (0,\infty)$.
%Each item $u$ is thus associated with a supply time $s_u$ and a perish or departure time $d_u$.
An item is \emph{present} if it has been supplied but has not yet perished or been discarded, whereas it is \emph{available} if it is both present and has not yet been sold. An item is discarded on arrival if the number of available items of the same good equals the seller's inventory capacity $C$ (unless otherwise specified, $C\rightarrow \infty$).
Buyers are unit-demand and arrive according to a Poisson process with rate $\gamma > 0$, with their types drawn i.i.d.~from a distribution $\calD$ over a set of $m$ types $\calB$, with type $j \in \calB$ bidding values $\mathbf{v}_j = (v_{ij})_{i \in \calG}$. Thus, buyers of type $j\in \calB$ arrive according to a Poisson process with rate $\gamma_j \triangleq \gamma \cdot \bbP_{\mathbf{v} \sim \calD} \left[\mathbf{v} = \mathbf{v}_j \right]$.
%and bid values $\mathbf{v} = (v_i)_{i \in \calG}$ upon arrival, where $\mathbf{v}$ is sampled independently from a discrete distribution $\calD$ with support $\{
%\mathbf{v}_1, \dots, \mathbf{v}_m\}$.\footnote{Equivalently, there is a fixed set $\calB$ of $m$ buyer types where each buyer type $j \in \calB$ has fixed (non-random) values $\mathbf{v}_j = (v_{ij})_{i \in \calG}$ and arrives according to a Poisson process with rate $\gamma_j \triangleq \gamma \cdot \bbP_{\mathbf{v} \sim \calD} \left[\mathbf{v} = \mathbf{v}_j \right]$.}
Upon the arrival of a buyer of type $j\in \calB$, the seller must irrevocably decide whether to sell an available item of at most one good $i\in \calG$ to the buyer at their bid price, $v_{ij}$, and the buyer immediately departs after the seller's decision.

An offline policy for an instance $\calI$ of the stationary prophet inequality problem knows in advance the realization of all the randomness of the input, i.e., 
it knows the times at which items of goods are supplied and perish, and the times at which different buyers arrive. 
An online policy, on the other hand, knows $\{\lambda_i,\, \mu_i\}_{i \in \calG}$ and $\{\mathbf{v}_j,\, \gamma_j\}_{j \in \calB}$ a priori, but does not know the realization of future randomness of the input.
An example of online single-good policies are \emph{posted-price policies}, which set a pair $(\bar v, \bar p)$ and accept all bids strictly greater than $\bar v$, accept bids equal to $\bar v$ with probability $\bar p$, and reject all bids strictly less than $\bar v$.
The optimal expected average reward of an unbounded-capacity offline (resp.,~online) policy for instance $\calI$ is denoted by $\optoff(\calI)$ (resp., $\opton(\calI)$). 
We measure online policies' average reward in terms of their approximation of $\optoff(\calI)$ and $\opton(\calI)$.

%Our single-good policies are \emph{posted-price policies}: we set a pair $(\bar v, \bar p)$ and accept all bids strictly greater than $\bar v$ with probability $1$, accepts bids equal to $\bar v$ with probability $\bar p$, and rejects all bids strictly less than $\bar v$.
%It is fairly easy to show that a randomized posted-price policy with average reward $R$ implies a deterministic posted-price policy with average reward $R/2$ (see \Cref{cla:randomized-rounded-to-deterministic}.)

\subsection{Prior LP benchmarks and a natural algorithm}
\label{sec:prior-LP-benchmarks}

%We consider three types of policies in decreasing order of strength: the optimal offline, the optimal online, and non-adaptive policies.
%An offline policy corresponds to a ``prophet'', who knows in advance the arrival and departure times of items and buyers, as well as the latter's bids.
%An online policy does not have advance knowledge of either arrival or departure times, nor of the buyers' bids.
%We denote the expected average gain of the optimal offline and online policies for instance $\calI$ of the stationary prophet inequality problem by $\optoff(\calI)$ and $\opton(\calI)$, respectively.
%Lastly, a non-adaptive (online) policy does not change over time or based on the state of the system.

\citet{collina2020dynamic} and \citet{aouad2020dynamic} present the following LP benchmarks, which upper bound the average gain of \emph{any} offline or online policy, respectively.

\begin{lem}[\cite{collina2020dynamic,aouad2020dynamic}]
\label{lem:basic-constraints}
Let $x_{ij}$ be the rate at which an offline policy sells items of good $i \in \calG$ to buyers of type $j \in \calB$. Then $\mathbf{x} = (x_{ij})_{i \in \calG, j \in \calB}$ satisfies the following constraints: 
\begin{align}
    \sum_{j \in \calB} x_{ij} & \leq \lambda_i & \forall i \in \calG \label{eqn:flow-constraint-seller} \\
    \sum_{i \in \calG} x_{ij} & \leq \gamma_j & \forall j \in \calB \label{eqn:flow-constraint-buyer} \\
    x_{ij} & \leq \gamma_j \cdot \frac{\lambda_i}{\mu_i} & \forall i \in \calG, j \in \calB \label{eqn:wine-20-constraint} \\ 
    x_{ij} & \geq 0. \label{eqn:positivity}
\end{align}
If $\mathbf{x}$ is the vector of rates derived by an \emph{online} policy, then $\mathbf{x}$ also satisfies the following constraint:
\begin{align}
    x_{ij} & \leq \gamma_j \cdot \left( \frac{\lambda_i - \sum_{\ell \in \calB} x_{i\ell}}{\mu_i} \right). \label{ec20-constraint}
\end{align}
\end{lem}

\begin{cor}%[\cite{collina2020dynamic,aouad2020dynamic}]
\label{cor:RB-geq-OPT}
For all instances $\calI$ of the stationary prophet inequality problem, 
\begin{enumerate}
    \item $\mathrm{RB}_\mathrm{off}(\calI) \triangleq \max\left\{\sum_{i, j} v_{ij} \cdot x_{ij} \mid \mathbf{x} \textrm{ satisfies } \eqref{eqn:flow-constraint-seller}-\eqref{eqn:positivity}\right\}$ satisfies $\mathrm{RB}_\mathrm{off}(\calI) \geq \optoff(\calI)$, and
    \item $\mathrm{RB}_\mathrm{on}(\calI) \triangleq \max\left\{\sum_{i,j} v_{ij} \cdot x_{ij} \mid \mathbf{x} \textrm{ satisfies } \eqref{eqn:flow-constraint-seller}-\eqref{ec20-constraint}\right\}$ satisfies $\mathrm{RB}_\mathrm{on}(\calI) \geq \opton(\calI)$.
    \end{enumerate}
\end{cor}

\citet{collina2020dynamic} present a simple non-adaptive algorithm for the stationary prophet inequality problem, 
which attempts to (approximately) follow the sale rates prescribed by a solution $\mathbf{x}^*$ to the offline benchmark $\mathrm{RB}_\mathrm{off}$, with additional parameters $\alpha\in [0,1]$ and $w_i$ satisfying $x_{ij}\leq \gamma_j\cdot w_i$. 
Their algorithm, generalized to take as input any $\mathbf{x}^*\in \mathbb{R}^{|\calG|\times |\calB|}$, is given in \Cref{alg:main}.\footnote{\label{collina-implied-bounds}The algorithm of \cite{collina2020dynamic} is \Cref{alg:main} with $\alpha=1$, applied to $\mathbf{x}^*$ an optimal solution to $\mathrm{RB}_\mathrm{off}$, with $w_i = \min\{1,\lambda_i / \mu_i\}$. 
The analysis of \cite{collina2020dynamic} implies that this algorithm is $\nicefrac{1}{3}$-competitive in the single-good setting, while we show that this algorithm is in fact $(1-\nicefrac{1}{(e-1)}) \approx 0.418$-competitive algorithm  (see \Cref{lem:WINE-alg-0.418-competitive}).
% In the multi-good setting, a simple modification of the analysis from \cite{collina2020dynamic} shows that their algorithm is $0.171$-competitive.
}
Intuitively, \Cref{alg:main} attempts to match good $i\in \calG$ and buyer type $j\in \calG$ at a rate close to $x^*_{ij}$. The role of high $\alpha$ is to increase the probability of a sale of good $i$ to buyer type $j$, conditioned on an item of good $i$ being available, while the role of a lower $\alpha$ is to increase the probability of items being available. The crux of the analysis is in bounding this probability, for which we rely on machinery from the queuing theory literature, described below.

%We refer the reader to \cite{collina2020dynamic} for a line-by-line discussion of \Cref{alg:main}, but briefly highlight the role of the (tunable) parameter $0 \leq \alpha \leq 1$ in \Cref{line:i-permits-sale}.
%Although $\alpha < 1$ means that the algorithm does not follow the sale rates from $\mathbf{x}^*$ quite as closely, the smaller $\alpha$ is, the more likely a buyer is to reach a good $i \in \calG$, regardless of where the good falls in the ordering from \Cref{line:j-reaches-i}. On the other hand, the larger $\alpha$ is, the more likely a buyer of type $j$ is to be sold an item of good $i$.\todo{next few lines giving away too much of what's to follow}
%Using an optimal solution to $\mathrm{RB}_\mathrm{off}$ and $w_i$ as above with $\alpha=1$ in \Cref{alg:main} can be shown to yield an $(1-\nicefrac{1}{(e-1)}) \approx 0.418$-approximation with respect to $\mathrm{RB}_\mathrm{off}$ in the single-good setting (see \Cref{lem:WINE-alg-0.418-competitive}) and a $0.171$-approximation in the multi-good setting.

\begin{algorithm}
    \begin{algorithmic}[1]
        %\State let $\mathbf{x^*}$ be an optimal solution to $\mathrm{RB}_\mathrm{off}(\calI)$ \label{line:LP-solution}
        \For{arrival of buyer of type $j \in \calB$}
            \For{each good $i \in \calG$ in a uniform random order} \label{line:j-reaches-i}
    		    \If{buyer unmatched and at least one item of good $i$ is available}
    		        \State sell with probability $p_
    		        {ij} \triangleq \alpha \cdot \frac{x_{ij}^*}{\gamma_j \cdot w_i}$ \label{line:i-permits-sale}
		        \EndIf
		    \EndFor
	    \EndFor
	\end{algorithmic}
    \caption{}
    \label{alg:main}
\end{algorithm}

\subsection{Queuing theory background}
Throughout, we will want to bound the probability of there being an available (i.e., present and unsold) item of good $i \in \calG$.
We denote this probability by $\bbP_C \left[A_i \geq 1 \right]$, where $A_i$ denotes the number of available items of good $i$.

\begin{restatable}{lem}{PrAvailableExact}
\label{lem:Pr-available-exact}
For any online policy with inventory capacity $C\in \Z_{>0}$ and which sells any available item of good $i \in \calG$ to buyers which arrive at rate $\gamma^*$, the stationary probability of an item of good $i\in \calG$ being available satisfies
\[ \bbP_C \left[ A_i \geq 1 \right] = 1 - \left(1 + \sum_{q=1}^C \prod_{r=1}^q \frac{\lambda_i}{r \cdot \mu_i + \gamma^*} \right)^{-1}. \]
\end{restatable}

As a corollary of the previous lemma, we have the following.

\begin{restatable}{cor}{PrAvailableBounds}
\label{cor:Pr-available-lower-upper-bounds}
For any online policy with inventory capacity $C\in \Z_{>0}$ and which sells any available item of good $i \in \calG$ to buyers which arrive at rate $\gamma^*$, the stationary probability of an item of good $i\in \calG$ being available satisfies
\[ \bbP_C \left[ A_i \geq 1 \right] \in \left[1 - \left(\sum_{q=0}^C \frac{1}{q!} \left( \frac{\lambda_i}{ \mu_i + \gamma^*} \right)^q \right)^{-1},\, 1 - \left(\sum_{q=0}^C \frac{1}{q!} \left( \frac{\lambda_i}{\mu_i} \right)^q \right)^{-1} \right]. \]
\end{restatable}
See \Cref{appendix:prelims} for the proofs of \Cref{lem:Pr-available-exact} and \Cref{cor:Pr-available-lower-upper-bounds}.
Notice that as $C$ approaches infinity (i.e., the seller has unbounded inventory capacity), the lower and upper bounds on $\bbP_C \left[A_i \geq 1 \right]$ approach $1 - \exp\left(-\lambda_i / (\mu_i + \gamma^*) \right)$ and $1 - \exp(-\lambda_i / \mu_i)$, respectively.

Finally, we recall the following fundamental PASTA property, due to \citet{wolff1982poisson}.
\begin{lem}[PASTA \cite{wolff1982poisson}]\label{pasta}
The fraction of Poisson arrivals who observe a stochastic process in a state is equal to the fraction of time the stochastic process is in this state, provided that the Poisson arrivals and the history of the stochastic process are independent.
\end{lem}

In our analysis of the multi-good problem we will need to prove stochastic dominance between two processes, for which the following lemma will prove useful. For this lemma, we recall that a set $S\subseteq \calY \subseteq \mathbb{R}^n$ is or upward closed if for every $y\geq \tilde y$ with $\tilde y \in S$ and $y\in \calY$, we have that $y \in S$.

\begin{lem}[\cite{brandt1994pathwise}]
\label{lem:stochastic-dominance}
Let $Y, \tilde Y$ be two stochastic processes taking values in $\calY\subseteq \mathbb{R}^n$, with time-homogeneous intensity matrices $Q, \tilde Q$.
Then, $Y$ stochastically dominates $\tilde Y$ ($\Pr[Y\geq y]\geq \Pr[ \tilde Y\geq y]$ for all $y\in \calY$) if and only if the following holds: for every $y, \tilde y \in \calY$ and upward closed set $S\subseteq \calY$, if $y\geq \tilde y$, and either $y, \tilde y \in S$ or $y, \tilde y \notin S$, then
\[ \sum_{z \in S} Q(y, z) \geq \sum_{z \in S} \tilde Q (\tilde y, z).\]
\end{lem}

\section{Tighter LP benchmarks}\label{sec:tighter-lps}

In this section, we present our tighter LP benchmarks for the optimal offline and online policies.
We start by noting that tighter benchmarks are needed than those considered in \cite{collina2020dynamic,aouad2020dynamic} in order to improve on prior approximations.
We then propose a new constraint which in fact tightens these benchmarks and enables us to obtain improved bounds in \Cref{sec:better-bounds}. 

\paragraph{Limitations of prior LP benchmarks.}
\label{sec:limitations-prior-LPs}
The algorithms of \citet{collina2020dynamic} and \citet{aouad2020dynamic} yield $(1-\nicefrac{1}{(e-1)})$-competitive and $(1 - \nicefrac{1}{e})$-approximate online policies in the single-good problem by relying on LP benchmarks $\mathrm{RB}_\mathrm{off}$ and $\mathrm{RB}_\mathrm{on}$.
Unfortunately, 
improving on the bounds implied by \cite{collina2020dynamic,aouad2020dynamic} for the stationary prophet inequality problem requires tighter benchmarks, since
these bounds are tight versus these LP benchmarks. (See proofs in \Cref{appendix:tighter-lps}.)

\begin{restatable}{obs}{RBoffgap}
\label{obs:RBoffgap}
There exist instances $\calI$ of the single-good stationary prophet inequality problem for which $\optoff(\calI) \leq (1-\nicefrac{1}{(e-1)})\cdot \mathrm{RB}_\mathrm{off}(\calI)$.
\end{restatable}

\begin{restatable}{obs}{RBongap}
\label{obs:RBongap}
There exist instances $\calI$ of the single-good stationary prophet inequality problem for which $\opton(\calI) \leq \left(1-\nicefrac{1}{e}\right)\cdot \mathrm{RB}_\mathrm{on}(\calI)$.
\end{restatable}

%The above observations, whose proof is deferred to \Cref{sec:omitted}, immediately rules out improved approximations of $\optoff(\calI)$ and $\opton(\calI)$ based on reward benchmarks $\mathrm{RB}_\mathrm{off}$ and $\mathrm{RB}_\mathrm{on}$.
%First, no policy\textemdash offline or online\textemdash is better than a $\nicefrac{(e-2)}{(e-1)}$-approximation of $\mathrm{RB}_\mathrm{off}$, meaning our analysis of the algorithm of \citet{collina2020dynamic} is tight in the single-good setting. Similarly, no online policy is better than a $(1-\nicefrac{1}{e})$-approximation of $\mathrm{RB}_\mathrm{on}$.
%Consequently, improving on the bounds implied by \cite{collina2020dynamic,aouad2020dynamic} for the stationary prophet inequality problem requires tighter benchmarks.

\subsection{A new constraint via the PASTA property \& new LP benchmarks}
\label{sec:new-constraint}

We introduce an additional constraint previously overlooked in the literature, which follows from the PASTA property, \Cref{pasta}.\footnote{%Following correspondence with the authors of \cite{collina2020dynamic}
In the third and latest arxiv version of \cite{collina2020dynamic} (uploaded January, 2021), the authors of that paper observe this constraint too, though they do not make use of it, instead replacing it by the potentially looser Constraint~\eqref{eqn:wine-20-constraint}.}
%\todo{for double-blind: rephrase this footnote?}

\begin{lem}
\label{lem:pasta-constraint}
Let $x_{ij}$ be the rate at which an offline (or online) policy sells items of good $i \in \calG$ to buyers of type $j \in \calB$.
Then $x_{ij}$ satisfies the following constraint:
\begin{equation}
\label{eqn:pasta-constraint}
    x_{ij} \leq \gamma_j \cdot \left(1 - \exp\left(-\lambda_i / \mu_i \right)\right).
\end{equation}
\end{lem}
\begin{proof}
The rate at which a policy sells items of good $i$ to buyers of type $j$ is trivially upper bounded by the rate at which such buyers arrive and inspect at least one present item.
The bound therefore follows from the PASTA property (\Cref{pasta}), together with the upper bound of \Cref{cor:Pr-available-lower-upper-bounds}.
\end{proof}
Combining \Cref{lem:pasta-constraint,lem:basic-constraints} yields the following tighter bounds on $\optoff(\calI)$ and $\opton(\calI)$.

\begin{cor}
\label{cor:LP-off-geq-OPT-off-LP-on-geq-OPT-on}
For all instances $\calI$ of the stationary prophet inequality problem, 
\begin{enumerate}
    \item $\mathrm{LP}_\mathrm{off}(\calI) \triangleq \max\left\{\sum_{i,j} v_{ij} \cdot x_{ij} \mid \mathbf{x} \textrm{ satisfies } \eqref{eqn:flow-constraint-seller}-\eqref{eqn:positivity}, \eqref{eqn:pasta-constraint} \right\}$ satisfies $\mathrm{LP}_\mathrm{off}(\calI) \geq \optoff(\calI)$, and
    \item $\mathrm{LP}_\mathrm{on}(\calI) \triangleq \max\left\{\sum_{i,j} v_{ij} \cdot x_{ij} \mid \mathbf{x} \textrm{ satisfies } \eqref{eqn:flow-constraint-seller}-\eqref{eqn:pasta-constraint} \right\}$ satisfies $\mathrm{LP}_\mathrm{on}(\calI) \geq \opton(\calI).$
\end{enumerate}
\end{cor}
% We note that Constraint~\eqref{eqn:pasta-constraint} subsumes Constraints~\eqref{eqn:flow-constraint-buyer} and \eqref{eqn:wine-20-constraint}, since $1-\exp(-z) \leq \min\{1,z\}$ for all $z\in \mathbb{R}$.
We note that Constraint~\eqref{eqn:pasta-constraint} subsumes Constraint~\eqref{eqn:wine-20-constraint} since $1-\exp(-z) \leq z$ for all $z\in \mathbb{R}$.
As we show in \Cref{sec:better-bounds}, this tighter constraint and the derived tighter LP bounds on $\optoff(\calI)$ and $\opton(\calI)$ allow for better approximations of the optimal offline and online policies for the single-good problem and a better competitive ratio in the multi-good problem.

\section{Improved approximation ratios}
%via \Cref{alg:main} with new LP benchmarks}
\label{sec:better-bounds}

In this section we present the proofs of our algorithmic results. 
In particular, we show that using the solutions to our tighter LP benchmarks $\mathrm{LP}_\mathrm{off/on}$ and an appropriate choice of $\alpha$ and $w_i$ for all $i \in \calG$ in \Cref{alg:main}, we achieve all of our improved approximation guarantees.

%We first consider the single-good problem.

\subsection{Single-good problem}
\label{sec:single-good}

When the seller has only one good $i$ for sale, we drop the subscript $i$ for notational convenience, using, e.g., $\lambda$ and $\mu$ as shorthand for $\lambda_i$ and $\mu_i$, and using $v_{j}$ and $x_j$ as shorthand for $v_{ij}$ and $x_{ij}$.
We assume, without loss of generality (due to rate re-scaling) that $\mu = 1$ and buyer types are sorted such that $v_1 > v_2 > \dots > v_m$.

%Furthermore, recall that although the inclusion of the $\alpha$ parameter when setting the permission probabilities in \Cref{alg:main} means that we deviate from the LP solution, the benefit is that it enables us to provide a better bound on the probability that a buyer of type $j \in \calB$ reaches each good $i \in \calG$, regardless of good $i$'s place in the ordering in \Cref{line:j-reaches-i}.
%However in the case of a single good, clearly any buyer, regardless of type or the value of $\alpha$, reaches the good in \Cref{line:j-reaches-i} with probability $1$, effectively nullifying the benefit of setting $\alpha < 1$.
%We therefore set $\alpha = 1$ throughout the remainder of the discussion on \Cref{alg:main} in the single-good problem.

We first note that for the single-good problem, \Cref{alg:main} with $\alpha=1$ (as we will use it) yields a \emph{posted-price} policy.

\begin{obs}
\label{obs:alg-posted-price}
For any instance $\calI$ of the single-good stationary prophet inequality problem, \Cref{alg:main} with $\alpha = 1$ is a posted-price policy if
\begin{enumerate}
    \item $\mathbf{x^*} \triangleq \arg\max_\mathbf{x} \mathrm{LP}_\mathrm{off}(\calI)$ and $w \triangleq 1 - \exp(-\lambda)$, or
    \item $\mathbf{x^*} \triangleq \arg\max_\mathbf{x} \mathrm{LP}_\mathrm{on}(\calI)$ and $w \triangleq \min\left\{1 - \exp(-\lambda), \lambda - \sum_{j \in \calB} x_j^*\right\}$.
\end{enumerate}
\end{obs}

\begin{proof}
By constraints~\eqref{eqn:positivity},  \eqref{ec20-constraint} and \eqref{eqn:pasta-constraint}, we have that $0\leq x^*_j\leq \gamma_j \cdot w$, and so these choices of $\mathbf{x}^*$ and $w$ guarantee that $p_j$ is a valid probability for all $j \in \calB$.
Furthermore, by local exchange arguments and the strict inequalities $v_1 > v_2 > \dots > v_m$, combined with $x^*_j \leq \gamma_j\cdot w$,
an optimal solution $\mathbf{x}^*$ to $\mathrm{LP}_\mathrm{off}$ for instance $\calI$ is in some sense greedy, and has the following form for some $\ell \in [m]$:
\[ x_j^* =
    \begin{cases}
        \gamma_j \cdot w & \text{for } j \leq \ell \\
        p _{\ell+1}\cdot \gamma_{\ell+1} \cdot w & \text{for } j = \ell + 1 \text{ and some } p_{\ell+1} \in [0, 1] \\
        0 & \text{for } j > \ell + 1.
    \end{cases}
\]
Therefore, \Cref{alg:main} with $\alpha=1$ and the above choices of $\mathbf{x}^*$ and $w$ is a posted-price policy characterized by $(v_{\ell+1}, p_{\ell+1})$.
\end{proof}

\paragraph{The sale rate.} By definition of \Cref{alg:main}, it is clear that an item of the single good $i$ is sold to a buyer of type $j$ if at least one item of good $i$ is available and the probabilistic test in \Cref{line:i-permits-sale} passes (which is independent of whether or not an item is available).
We say a buyer who passes this probabilistic test of \Cref{line:i-permits-sale} has his bid \emph{permitted} by the seller.
By standard Poisson splitting, buyers of type $j$ whose bid is permitted arrive at rate $\gamma_j\cdot p_j = x^*_j /w$, and by Poisson merging, the arrival rate of buyers (of any type) whose bid the seller permits is
$\gamma^* \triangleq \sum_{j\in \calB} \gamma_j \cdot p_j = \sum_{j \in \calB} x_j^* / w$.
From the above, by the PASTA property (\Cref{pasta}), we have the following expression for the selling rate to buyers of type $j$ for our policy with capacity $C$, which we denote by $s_j^C$:
\begin{equation}
\label{eqn:s-ij-C-single-good}
    s_j^C =  \gamma_j\cdot p_j\cdot \bbP_C[A\geq 1] = \frac{\bbP_C \left[ A \geq 1 \right]}{w} \cdot x_j^*,
\end{equation}
where $A$ is the number of available items on arrival of the buyer.
%By \Cref{lem:Pr-available-exact}, we have that
%\begin{equation}
%\label{eqn:Pr-available-single-good}
%    \bbP_C\left[ A \geq 1 \right] = 1 - \left(1 + \sum_{q=1}^C \prod_{r=1}^q \frac{\lambda}{r + \gamma^*} \right)^{-1}.
%\end{equation}
In what follows, we will obtain our algorithmic results of \Cref{thm:1/2-competitive-single-good-policy,thm:0.656-approximate-single-good-policy} by lower bounding $\bbP_C[A \geq 1]/w$, from which our approximation ratios follow by linearity of expectation.

\subsubsection{Proof of \Cref{thm:1/2-competitive-single-good-policy} (algorithmic result)}
\label{sec:single-good-vs-opt-off}

In this section we prove the algorithmic result of \Cref{thm:1/2-competitive-single-good-policy}, restated below.
%\HalfCompetitiveSingleGoodPolicy*
\begin{lem}\label{half-competitive-restated}
There exists a $\nicefrac{1}{2}$-competitive posted-price policy with capacity $C=2$ for the single-good stationary prophet inequality problem.
\end{lem}
\begin{proof}
Fix an instance $\calI$ of the single-good stationary prophet inequality problem where the seller has inventory capacity $C = 2$.
Consider \Cref{alg:main} where $\mathbf{x}^* = \arg\max_\mathbf{x} \mathrm{LP}_\mathrm{off}(\calI)$, $w \triangleq 1 - \exp(-\lambda)$, and $\alpha = 1$, which is indeed a posted-price policy by \Cref{obs:alg-posted-price}.
By definition of $w$ and constraint~\eqref{eqn:flow-constraint-seller}, we have that
\begin{equation}
\label{eqn:gamma*-leq-lambda-divided-by-m*}
    \gamma^* = \sum_{j\in \calB} p_j \cdot \gamma_j = \frac{\sum_{j \in \calB} x_j^*}{1 - \exp(-\lambda)} \leq \frac{\lambda}{1 - \exp(-\lambda)}.
\end{equation}
By \Cref{lem:Pr-available-exact} and \Cref{eqn:gamma*-leq-lambda-divided-by-m*}, we therefore have that
%\Cref{eqn:Pr-available-single-good} that
\begin{align}
\label{eqn:Pr-available-divided-by-m*-geq-1/2}
    \frac{\bbP_2 \left[ A \geq 1 \right]}{w} =  \frac{1 - \left(1 + \sum_{q=1}^2 \prod_{r=1}^q \frac{\lambda}{r + \gamma^*} \right)^{-1}}{1 - \exp(-\lambda)}
    \geq & \frac{1 - \left(1 + \sum_{q=1}^2 \prod_{r=1}^q \frac{\lambda}{r + \frac{\lambda}{1 - \exp(-\lambda)}} \right)^{-1}}{1 - \exp(-\lambda)} \geq \frac{1}{2},
%    \geq & \min_{z > 0} \frac{1 - \left(1 + \sum_{q=1}^C \prod_{r=1}^q \frac{z}{r + \frac{z}{1 - \exp(-z)}} \right)^{-1}}{1 - \exp(-z)}.\nonumber  \\
%    = & \frac{1 - \left(1 + \sum_{q=1}^2 \prod_{r=1}^q \frac{z}{r + \frac{z}{1 - \exp(-z)}} \right)^{-1}}{1 - \exp(-z)} \geq \frac{1}{2},
\end{align}
%where 
%\[ f(C) \triangleq \min_{z > 0} \frac{1 - \left(1 + \sum_{q=1}^C \prod_{r=1}^q \frac{z}{r + \frac{z}{1 - \exp(-z)}} \right)^{-1}}{1 - \exp(-z)}. \]
where the last inequality holds for all $\lambda\geq 0$; this inequality is easily verified to hold (with equality) for $\lambda\rightarrow0^+$, while the left-hand side of the inequality can be shown (with some effort) to be monotone increasing in $\lambda\geq 0$, and therefore this inequality holds for all $\lambda\geq 0$.
We defer a formal proof of this last inequality, restated in the following claim, to \Cref{appendix:better-bounds}.
%verified by deriving $\frac{1 - \left(1 + \sum_{q=1}^C \prod_{r=1}^q \frac{\lambda}{r + \frac{\lambda}{1 - \exp(-\lambda)}} \right)^{-1}}{1 - \exp(-\lambda)}$ in terms of $\lambda$

%relied on the following inequality, whose proof we defer to \Cref{appendix:better-bounds}.

\begin{restatable}{cla}{LowerBoundOneHalf}
\label{cla:lower-bound-one-half}
For all $x \in \mathbb{R}_{\geq 0}$, we have that $\frac{1 - \left(1 + \sum_{q=1}^2 \prod_{r=1}^q \frac{x}{r + \cdot \frac{x}{1 - \exp(-x)}} \right)^{-1}}{1 - \exp(-x)} \geq \frac{1}{2}.$
\end{restatable}

%Since $f$ is monotonically increasing in $C$, $\lim_{C \rightarrow \infty} f(C) \geq f(2) \geq 1/2$, where the last inequality holds by \Cref{cla:lower-bound-one-half}, the proof of which we defer to \Cref{appendix:better-bounds}.
Combining \Cref{eqn:s-ij-C-single-good,eqn:Pr-available-divided-by-m*-geq-1/2} yields $s_j^2 \geq \nicefrac{1}{2} \cdot x_j^*$ for all buyer types $j \in \calB$.
We conclude that, by linearity of expectation and \Cref{cor:LP-off-geq-OPT-off-LP-on-geq-OPT-on}, the expected average revenue of \Cref{alg:main} for $\calI$ is at least half that of the optimal offline policy in expectation.
\end{proof}

Note that in our proof of \Cref{half-competitive-restated} we used a policy with capacity $C=2$. Using the same proof strategy, the same policy with further restricted capacity of $C=1$ can be shown to yield a competitive ratio of $\approx 0.435$ (see also \Cref{cor:0.435-competitive-single-good-WINE-analysis-our-LP} from \Cref{appendix:priorwork}).
On the other hand, monotonicity of $\bbP_C[A\geq 1]$ as a function of $C$ implies that \Cref{eqn:Pr-available-divided-by-m*-geq-1/2} holds for any capacity $C\geq 2$. Intuitively, such higher capacity should allow for strictly higher competitive ratio. However, this turns out to not be the case; in \Cref{sec:impossibility} we show that no online policy, regardless of inventory capacity and even not restricted to posted prices, is better than $\nicefrac{1}{2}$-competitive, and so this capacity-two posted-price policy is optimally competitive among all online policies. 

%\begin{rem}
%The proof of \Cref{thm:1/2-competitive-single-good-policy,cla:lower-bound-one-half} together imply that inventory capacity of size $2$ or greater is sufficient to guarantee that \Cref{alg:main} a $\nicefrac{1}{2}$-competitive posted-price policy.
%\end{rem}

%We show in \Cref{sec:impossibility} that no online policy, regardless of inventory capacity, is actually better than a $\nicefrac{1}{2}$-approximation of the optimal offline policy. 
% \Cref{thm:1/2-competitive-single-good-policy,cla:lower-bound-one-half} together imply that inventory capacity of size $2$ or greater is sufficient to guarantee that \Cref{alg:main} $\nicefrac{1}{2}$-competitive posted-price policy, which we show later is optimal among \emph{all} online policies.

\subsubsection{Proof of \Cref{thm:0.656-approximate-single-good-policy}}
\label{sec:single-good-vs-opt-on}

In this section we prove our results for approximating $\opton(\calI)$ for the single-good problem, restated below.

%We start with a brief warm-up, achieving $1-\nicefrac{1}{e}$, using our single-bid policy. 

\ApproximateSingleGoodPolicy*
\begin{proof}
Fix an instance $\calI$ of the single-good stationary prophet inequality problem where the seller has inventory capacity $C \in \Z_{> 0}$.
Consider \Cref{alg:main} where $\mathbf{x^*} = \arg\max_\mathbf{x} \mathrm{LP}_\mathrm{on}(\calI)$, $\alpha = 1$, and $w \triangleq \min\left\{1 - \exp(-\lambda),\, \lambda - \sum_{j \in \calB} x_j^* \right\}$, which is a posted-price policy by \Cref{obs:alg-posted-price}.
By \Cref{eqn:s-ij-C-single-good} and the linearity of expectation, our policy's approximation ratio is at least $\bbP_C[A\geq 1] / w$, where, as before, $A$ is the number of available items.
We therefore turn to lower bounding $\bbP_C[A\geq 1] / w.$

For this, we will invoke our lower bound on $\bbP_C[A\geq 1]$ of \Cref{cor:Pr-available-lower-upper-bounds}, which requires upper bounds on $\gamma^*$---the rate at which buyers arrive and pass the probabilistic check in \Cref{line:i-permits-sale}.
By definition of $w$, we have the following bound on $\gamma^*$:
\begin{equation}
\label{eqn:1-plus-gamma*-leq-lambda-divided-by-m*}
    1 + \gamma^* = \frac{w + \sum_{j \in \calB} x^*_j}{w} \leq \frac{\lambda - \sum_{j \in \calB} x_j^* + \sum_{j \in \calB} x^*_j}{w} = \frac{\lambda}{w}.
\end{equation}

\paragraph{Warm-up $1-\nicefrac{1}{e}$ bound:} Using the above bound on $1+\gamma^*$ and appealing to the lower bound of \Cref{cor:Pr-available-lower-upper-bounds}, we find that our policy with $C=\infty$ has approximation ratio at least $1-\nicefrac{1}{e}$, since for any $w\in [0,1]$ (as is the case for our $w\leq 1-\exp(-\lambda)\leq 1$), we have that
$$\frac{\bbP_\infty[A\geq 1]}{w} \geq \frac{1 - \left(\sum_{q=0}^\infty \frac{1}{q!} \left( \frac{\lambda}{ 1 + \gamma^*} \right)^q \right)^{-1}}{w} \geq \frac{1 - \left(\sum_{q=0}^\infty \frac{w^q}{q!} \right)^{-1}}{w}  = \frac{1-\exp(-w)}{w} \geq 1-\nicefrac{1}{e}.$$
In order to improve on the above natural bound, we will rely on the following two claims.

\begin{cla}\label{cla:Pr-available-m*-leq}
    If $w\leq 1-\exp(-\nicefrac{12}{5})$, then, for $g_1(C, x) \triangleq \frac{1 - \left(\sum_{q=0}^C \frac{x^q}{q!} \right)^{-1}}{x}$, we have that $$\frac{\bbP_C[A\geq 1]}{w} \geq g_1(C, 1 - \exp(-\nicefrac{12}{5})).$$
\end{cla}
\begin{proof}
    By \Cref{cor:Pr-available-lower-upper-bounds} and \Cref{eqn:1-plus-gamma*-leq-lambda-divided-by-m*}, we have $\frac{\bbP_C[A\geq 1]}{w}\geq \frac{\Big(1 - \left(\sum_{q=0}^C \frac{w^q}{q!} \right)^{-1}\Big)}{w} = g_1(C,w)$. The claim then follows since $g_1(C,w)$ is monotone decreasing as a function of $w$ (see \Cref{monotonicity-of-g1}).
\end{proof}

The more intricate claim is the following bound on $\bbP_C[A\geq 1]/w$ for $w$ close to $1$.

\begin{cla}
\label{cla:Pr-available-m*-geq}
If $w > 1 - \exp\left(-\nicefrac{12}{5}\right)$, then, for $g_2(C,x)\triangleq \frac{1 - \left(\frac{1}{6} + \frac{5}{6} \cdot \sum_{q=0}^C \frac{1}{q!} \left( \frac{6 \cdot x}{5} \right)^q \right)^{-1}}{x}$, we have that $$\frac{\bbP_C \left[ A \geq 1 \right]}{w} \geq g_2(C,1).$$
\end{cla}
% Mathematica: g2[C_, w_] := (1 - (1/6 + 5/6*Sum[1/q!*(6*w/5)^q, {q, 0, C}])^(-1))/w

\begin{proof}
Observe that the lower bound on $\mathbb{P}_C \left[ A \geq 1 \right]$ based on \Cref{cor:Pr-available-lower-upper-bounds} that we used in \Cref{cla:Pr-available-m*-leq} follows by noting that $r + \gamma^* \leq r \cdot (1 + \gamma^*)$ for all $r \geq 1$, which is lossy for large $r$. In particular, for $r \geq 2$ we have
\[ r - 1 \leq 2 \cdot r - \frac{12}{5} = \frac{12}{5} \cdot \left(\frac{5}{6} \cdot r - 1 \right) \leq \frac{\lambda}{w} \cdot \left(\frac{5}{6} \cdot r - 1 \right), \]
where the second inequality follows from the fact that $w \geq 1 - \exp\left(-\nicefrac{12}{5}\right)$ implies $\lambda \geq \nicefrac{12}{5}$ and the fact that $w \leq 1-\exp(-\lambda)\leq 1$.
This, combined with \Cref{eqn:1-plus-gamma*-leq-lambda-divided-by-m*} yields
\[ r + \gamma^* = r - 1 + 1 + \gamma^* \leq \frac{\lambda}{w} \cdot \left(\frac{5}{6} \cdot r - 1 \right) + \frac{\lambda}{w} = \frac{5}{6} \cdot r \cdot \frac{\lambda}{w}. \]
Rearranging terms, we thus have that for all $r \geq 2$, 
\begin{equation}
\label{eqn:lambda-divided-by-j+gamma*-geq}
    \frac{\lambda}{r + \gamma^*} \geq \frac{6}{5} \cdot \frac{w}{r},
\end{equation}
This combined with \Cref{lem:Pr-available-exact} and \Cref{eqn:gamma*-leq-lambda-divided-by-m*} yields
%\Cref{eqn:Pr-available-single-good} that
\begin{align}\nonumber
    \frac{\bbP_C \left[ A \geq 1 \right]}{w} =  \frac{1 - \left(1 + \sum_{q=1}^C \prod_{r=1}^q \frac{\lambda}{r + \gamma^*} \right)^{-1}}{1 - \exp(-\lambda)}
    \geq \frac{1 - \left(\frac{1}{6} + \frac{5}{6} \cdot \sum_{q=0}^C \frac{1}{q!} \left( \frac{6 \cdot w}{5} \right)^q \right)^{-1}}{w} = g_2(C, w).
\end{align}
The claim then follows since $g_2(C,w)$ is monotone decreasing in $w$ (see \Cref{monotonicity-of-g2}).
\begin{comment}

and \Cref{eqn:Pr-available-single-good,eqn:lambda-divided-by-j+gamma*-geq} together yield
\[ \frac{\bbP_C \left[ A \geq 1 \right]}{w} \geq \frac{1 - \left(\frac{1}{6} + \frac{5}{6} \cdot \sum_{q=0}^C \frac{1}{q!} \left( \frac{6 \cdot w}{5} \right)^q \right)^{-1}}{w} \triangleq g_2(C, w). \]
Therefore, $\bbP_C \left[ A \geq 1 \right] / w \geq \min\left\{ g_2(C, w) : w \in [1 - \exp(-12/5), 1] \right\}$.

\begin{restatable}{fact}{gTwoLowerBound}
\label{fact:g2-lower-bound}
For all $C \in \Z_{> 0 }$ and $w \in [1 - \exp(-\nicefrac{12}{5}), 1] $, we have that $g_2(C, w) \geq g_2(C, 1)$.
\end{restatable}

See \Cref{appendix:better-bounds} for the proof of \Cref{fact:g2-lower-bound}.
\end{comment}
\end{proof}

\begin{comment}
Observe that $\lim_{C \rightarrow \infty} g_1(C, w) = \left(1 - \exp(-w) \right) / w$, and therefore by \Cref{fact:g1-lower-bound}, for any $w \in [0, 1 - \exp(-12/5)]$,
\[ \lim_{C \rightarrow \infty} g_1(C, w) \geq \lim_{C \rightarrow \infty} g_1(C, 1 - \exp(-12/5)) = \frac{1 - \exp\left(- \left(1 - \exp(-12/5) \right) \right)}{1 - \exp(-12/5)} \geq 0.656. \]
Similarly, $\lim_{C \rightarrow \infty} g_2(C,w) = \left(1 - \left(\frac{1}{6}+\frac{5}{6} \cdot \exp\left( \frac{6 \cdot w}{5} \right)\right)^{-1}\right) / w$, and therefore by \Cref{fact:g2-lower-bound}, for any $w \in [1 - \exp(-12/5), 1]$,
\[ \lim_{C \rightarrow \infty} g_2(C, w) \geq \lim_{C \rightarrow \infty} g_2(C, 1) = 1 - \left(\frac{1}{6}+\frac{5}{6} \cdot \exp\left( \frac{6}{5} \right)\right)^{-1} \geq 0.656. \qedhere \]
\end{comment}

Combining \Cref{eqn:s-ij-C-single-good} and \Cref{cla:Pr-available-m*-leq,cla:Pr-available-m*-geq},
we obtain the following:
\begin{align}\label{approx-opton-as-fn-of-C}
    s_j^C \geq \min\{g_1(C,1-\exp(-\nicefrac{12}{5}), g_2(C,1)\}\cdot x^*_j.
\end{align}
The above bound is at least $0.656$ for all $C\geq 5$.
By linearity of expectation, we conclude that our policy's average gain is at least a $0.656$ fraction of the average gain of the optimal online policy.
\end{proof}

\paragraph{A small (but large enough) inventory suffices.} 
Our proof of \Cref{thm:0.656-approximate-single-good-policy} shows that a capacity of $C=5$ is sufficient to obtain a $0.656$ approximation of $\opton(\calI)$.
Opening up this proof and evaluating \Cref{approx-opton-as-fn-of-C} for different $C$, we obtain a number of bounds for different inventory capacity $C$; e.g., a capacity of $C=3$ is sufficient to break the natural barrier of $1-\nicefrac{1}{e}$ for this problem, and a capacity of $C=1$ yields a $\nicefrac{1}{2}$ approximation.
In contrast, we also prove in \Cref{cla:upper-bound-approx-ratio-inventory-C} (proof deferred to \Cref{appendix:better-bounds}) that sufficiently small inventory does harm the approximation ratio; so, for example, our bound for $C=1$ is optimal.
See \Cref{table:approximations-for-C-bounded-inventory} for comparison of upper and lower bounds in terms of $C$.

\begin{comment}
Since $g_1(C,w)$, $g_2(C,w)$ are easily seen to be monotonically increasing in $C$, and since by \Cref{fact:g1-lower-bound,fact:g2-lower-bound}, we have that for any $C \geq 5$,
\[ \frac{\bbP_C [ A \geq 1 ]}{w} \geq \min\left\{g_1(5, 1-\exp(-12/5)), g_2(5, 1) \right\} \geq 0.656. \]
That is, inventory capacity of size $5$ or greater is sufficient to guarantee that the approximation factor from \Cref{thm:0.656-approximate-single-good-policy} still holds.
In contrast, we also prove in \Cref{cla:upper-bound-approx-ratio-inventory-C} (proof deferred to \Cref{appendix:better-bounds}) that sufficiently small inventory does harm the approximation ratio.
Evaluating $\min\left\{g_1(C, 1-\exp(-12/5)), g_2(C, 1) \right\}$ for smaller values of $C$ yields approximations for more constrained inventory capacities as well\textemdash see \Cref{table:approximations-for-C-bounded-inventory} for comparison of upper and lower bounds in terms of $C$.
\end{comment}

\begin{restatable}{lem}{UpperBoundApproxRatioInventoryC}
\label{cla:upper-bound-approx-ratio-inventory-C}
For any $C \in \Z_{> 0}$, no online policy with inventory capacity $C$ is greater than a $\nicefrac{C}{(C+1)}$-approximation of the optimal online (unbounded inventory) policy.
\end{restatable}

\begin{table}[h]
\centering
\begin{tabular}{|c|c|c|c|c|c|}
\hline
& \multicolumn{5}{c|}{Inventory size $C$} \\ \cline{2-6}
& $1$ & $2$ & $3$ & $4$ & $5$ \\ \hline
Lower bound (algorithm) & $\nicefrac{1}{2}$ & $0.615$ & $0.647$ & $0.655$ & $0.656$ \\ \hline
Upper bound (impossibility) & $\nicefrac{1}{2}$ & $\nicefrac{2}{3}$ & $\nicefrac{3}{4}$ & $\nicefrac{4}{5}$ & $\nicefrac{5}{6}$ \\ \hline
\end{tabular}
\caption{Lower and upper bounds on the approximation factor of the policy from \Cref{thm:0.656-approximate-single-good-policy} relative to the optimal online policy when the seller has inventory capacity $C$.}
\label{table:approximations-for-C-bounded-inventory}
\end{table}

\subsection{Multi-good problem}
\label{sec:multi-good}

In this section we analyze \Cref{alg:main} applied to a multi-good instance.

We start by introducing the following terminology used in this section's analysis of \Cref{alg:main}.
We say that a buyer of type $j \in \calB$ \emph{reaches} good $i \in \calG$ if the buyer has not yet been sold an item when he considers good $i$ in \Cref{line:j-reaches-i} of \Cref{alg:main}, and we denote this event by $R_{ij}$
We denote the number of available items of good $i$ by $A_i$.
Lastly, we say that the seller \emph{permits} the sale of an item of good $i$ to a buyer of type $j$ if the probabilistic check on \Cref{line:i-permits-sale} would pass, without regard for whether or not the buyer has already been sold an item of another good or the availability of the good.
We denote this event by $P_{ij}\sim Ber(p_{ij})$.
Importantly, the seller permits a sale independently of the availability of any good.

By definition of \Cref{alg:main}, it is clear that an item of good $i$ is sold to a buyer of type $j$ if the buyer reaches good $i$, at least one item of good $i$ is available, and the seller permits the sale of good $i$ to the buyer.
Suppose the seller has inventory capacity $C \in \Z_{> 0}$, and let $s_{ij}^C$ denote the expected rate at which items of good $i$ are sold to buyers of type $j \in \calB$ under \Cref{alg:main}.
By the PASTA property (\Cref{pasta}), we have the following:
\begin{equation}
\label{eqn:s-ij-C}
    s_{ij}^C = \gamma_j \cdot \bbP_C \left[ 
    R_{ij} \wedge A_i \geq 1 \wedge P_{ij} \right] = \alpha \cdot \frac{\bbP_C \left[ R_{ij} \wedge A_i \geq 1 \right]}{w_i} \cdot x_{ij}^*.
\end{equation}

%For the multi-good problem, it is not necessarily the case that a buyer of type $j \in \calB$ reaches good $i \in \calG$, since this buyer may be sold an item of good $i^\prime$ for some $i^\prime$ that precedes $i$ in the random ordering from \Cref{line:j-reaches-i} in \Cref{alg:main}.
%Furthermore, whether a buyer reaches a good and its availability are correlated via the availability of other goods, which makes lower bounding $\bbP_C \left[ R_{ij} \wedge A_i \geq 1 \right]$ from \Cref{eqn:s-ij-C} more challenging in the multi-good than in the single-good problem.

Inspecting \Cref{eqn:s-ij-C}, it is clear that lower bounding $s_{ij}^C$ is more challenging in the multi-good problem than in the single good, as the event that a buyer of type $j$ reaches a good $i\in \calG$ and the availability of good $i$ are correlated via the availability of other goods.
\citet{collina2020dynamic} approached this problem by considering a number of Poisson processes which they showed to either stochastically dominate or be dominated by the processes of interest, and further, by arguing about the correlations between them.
Although we similarly introduce a stochastic dominance relation, our analysis also relies heavily on the same queuing theory techniques employed in the single-good problem for analyzing stationary distributions of CTMCs.
% \todo{\tiny contrast our approach with \cite{collina2020dynamic}}

To lower bound $\bbP_C[R_{ij} \wedge A_i\geq 1]$, we introduce some additional notation.
We denote by $\tilde R_{ij}$ the event that a buyer of type $j$ is not permitted by the seller to buy a good $i'\in \calG$ that is present (even if unavailable) and precedes good $i$ in the random ordering from \Cref{line:j-reaches-i}.
Although similar to $R_{ij}$, observe that $\tilde R_{ij}$ does not depend on the availability of any good.
We also let $\tilde A_i$ denote the number of items of good $i$ available when $i$ comes first in the ordering from \Cref{line:j-reaches-i}.
We bound the probability that a buyer of type $j$ reaches good $i$ and the good is available in terms of $\tilde R_{ij}$ and $\tilde A_i$ in the following claim.

\begin{restatable}{lem}{PrReachesAndAvailable}
\label{cla:Pr-reaches-and-available-lower-bound}
For any good $i \in \calG$ and buyer type $j \in \calB$,
\begin{equation}
\label{eq:R-ij-A-i-geq-tilde-R-ij-tilde-A-i}
    \bbP_C \left[ R_{ij} \wedge A_i \geq 1 \right] \geq \bbP_C \left[ \tilde R_{ij} \right] \cdot \bbP_C \left[ \tilde A_i \geq 1 \right].
\end{equation} 
\end{restatable}

See \Cref{appendix:multi-good} for the proof, which proceeds by introducing a new stochastic process which relaxes the constraints of the multi-good problem so that, in some sense, each good (simultaneously) comes first in the random ordering.
For each good $i \in \calG$, this only increases the probability that an item of good $i$ is sold to an arriving buyer, thus creating more ``downwards pressure'' on the number of items available.
It is intuitive, therefore, that the dynamics under \Cref{alg:main} stochastically dominate those of this new process, and as a result, we can lower bound the probability that a buyer reaches a good and there is an item of that good available under \Cref{alg:main} by considering the probability of these same events under the dominanted stochastic process.
By construction, these events, which we relate to $\tilde R_{ij}$ and $\tilde A_i \geq 1$, are independent under this new process, and we can therefore lower bound each separately.

In particular, lower bounding $\bbP_C \left[ \tilde R_{ij} \right]$, as in the following lemma (see \Cref{appendix:multi-good} for proof), is straightforward as it depends only on the presence, not availability, of goods, which we can characterize exactly from the upper bound of \Cref{cor:Pr-available-lower-upper-bounds}.

% We now turn to bounding $\bbP_C \left[ \tilde R_{ij} \right]$ and $\bbP_C \left[ \tilde A_i \geq 1 \right]$.
% The former is straightforward since, as previously discussed, it depends only on the presence, not availability, of goods, which we can characterize exactly from the upper bound of \Cref{cor:Pr-available-lower-upper-bounds}.
% The latter is also straightforward in light of our analysis of the single-good problem.
% As previously mentioned, the availability of good $i \in \calG$ under $\tilde Q_C$ reduces to the availability of the good in the single-good problem (see the proof of \Cref{thm:1/2-competitive-single-good-policy}), and the proof is therefore omitted.

\begin{restatable}[\cite{collina2020dynamic}]{lem}{PrJReachesI}
\label{lem:Pr-j-reaches-i}
For any good $i \in \calG$ and buyer type $j \in \calB$, we have that
\[ \bbP_C \left[ \tilde R_{ij} \right] \geq 1 - \frac{\alpha}{2}. \]
\end{restatable}

Combining \Cref{cla:Pr-reaches-and-available-lower-bound,lem:Pr-j-reaches-i}, we have the following theorem regarding the competitive ratio of \Cref{alg:main} in the multi-good problem. 

\CompetitiveMultiGoodPolicy*
\begin{proof}
Fix an instance $\calI$ of the multi-good stationary
%\todo{\tiny mathematica says this is optimized by taking $\alpha=\sqrt{3}-1$, giving ratio $-(((-3 + Sqrt[3]) (-1 + Sqrt[3]))/(2 Sqrt[3]))\approx 0.2679$, versus $0.2678$. Good thing we stick to simpler $\alpha$, and keep it simple} 
prophet inequality problem where the seller has inventory capacity $C = 2$ and consider \Cref{alg:main} where $\mathbf{x^*} = \arg\max_\mathbf{x} \mathrm{LP}_\mathrm{off}(\calI)$ and $w_i \triangleq 1 - \exp(-\lambda_i / \mu_i)$ for all $i \in \calG$.
For any good $i \in \calG$ and buyer type $j \in \calB$, it follows from \Cref{eqn:s-ij-C,cla:Pr-reaches-and-available-lower-bound} that the expected rate at which items of good $i$ are sold to buyers of type $j$ under \Cref{alg:main} is at least
\[
    s_{ij}^2 \geq \alpha \cdot \left(1 - \frac{\alpha}{2}\right) \cdot \frac{\bbP_2 \left[ \tilde A_i \geq 1 \right]}{w_i} \cdot x_{ij}^* \geq \alpha \cdot \left(1 - \frac{\alpha}{2}\right) \cdot \frac{1 - \left(1 + \sum_{q=1}^2 \prod_{r=1}^q \frac{\lambda_i / \mu_i}{r + \alpha \cdot \frac{\lambda_i / \mu_i}{1 - \exp\left(-\lambda_i / \mu_i  \right)}} \right)^{-1}}{1 - \exp\left(-\lambda_i / \mu_i \right)} \cdot x_{ij}^*,
\]
where the second inequality holds due to \Cref{eqn:Pr-available-divided-by-m*-geq-1/2}, as the probability that an item of good $i$ is available when $i$ comes first in the ordering (i.e., $\tilde{A_i}\geq 1$) is exactly the probability that an item of good $i$ is available in the single-good problem.
For $\alpha = 3/4$, we have that
\[
    s_{ij}^2 \geq \frac{3}{4} \cdot \left(1 - \frac{3}{8} \right) \cdot \frac{4}{7} \cdot x_{ij}^* = \frac{15}{56} \cdot x_{ij}^*.
\]
(The proof of the $4/7$ lower bound is analogous to that of \Cref{cla:lower-bound-one-half}, and is omitted.)
\begin{comment}
\[ h(C, \alpha) \triangleq \min_{z > 0} \frac{1 - \left(1 + \sum_{q=1}^C \prod_{r=1}^q \frac{z}{r + \alpha \cdot \frac{z}{1 - \exp(-z)}} \right)^{-1}}{1 - \exp(-z)}. \]
For any $0 \leq \alpha \leq 1$, $h$ is monotonically increasing in $C$, so $\lim_{C \rightarrow \infty} h(C, \alpha) \geq h(2, \alpha)$.

Letting $\alpha = 3/4$, we have the following claim.

\begin{restatable}{cla}{LowerBoundCompetitiveRatioMultiGood}
\label{cla:h-lower-bound}
For all $z \in \mathbb{R}_{> 0}$,
\[ 1 - \left(1 + \sum_{q=1}^2 \prod_{r=1}^q \frac{z}{r + \frac{3}{4} \cdot \frac{z}{1 - \exp(-z)}} \right)^{-1} \geq \frac{4}{7} \cdot \left( 1 - \exp(-z) \right). \]
\end{restatable}
\end{comment}
%\Cref{cla:h-lower-bound}, the proof of which we defer to \Cref{appendix:better-bounds}, and 
% Equation~\eqref{eqn:s-ij-C-multi-good} therefore yields
% \[ s_{ij}^C \geq \frac{3}{4} \cdot \left( 1 - \frac{3}{8}\right) \cdot \frac{4}{7} \cdot x_{ij}^* = \frac{15}{56} \cdot x_{ij}^*.\]
%for all $C \geq 2$.
It follows from linearity of expectation and \Cref{cor:LP-off-geq-OPT-off-LP-on-geq-OPT-on} that the expected average reward of \Cref{alg:main} for $\calI$ is at least $\nicefrac{15}{56}$ that of the optimal offline policy in expectation.
% Notice that when the seller has unbounded inventory capacity, we have
% \[ \lim_{C \rightarrow \infty} s_{ij}^C \geq \alpha \cdot \left(1 - \frac{\alpha}{2}\right) \cdot \frac{1 - \exp\left(-\frac{\lambda_i}{\mu_i + \alpha \cdot \frac{\lambda_i}{1 - \exp\left(-\lambda_i / \mu_i  \right)}} \right)}{1 - \exp\left(-\lambda_i / \mu_i \right)} \cdot x_{ij}^* \geq \alpha \cdot \left(1 - \frac{\alpha}{2} \right) \cdot \left( \min_{z>0} h(\alpha, z) \right) \cdot x_{ij}^*, \]
% where
% \[ h(\alpha, z) \triangleq \frac{1 - \exp\left(-\frac{z}{1 + \alpha \cdot \frac{z}{1 - \exp(-z)}} \right)}{1 - \exp(-z)}. \]
% For fixed $0 \leq \alpha \leq 1$, $h(\alpha, z)$ is monotonically increasing in $z$\todo{Necessary to prove?} and therefore $\min_{z>0} h(\alpha, z) = \lim_{z \rightarrow 0} h(\alpha, z) = 1 / (1 + \alpha)$.
% As a result,
% \[ \lim_{C \rightarrow \infty} s_{ij} \geq \alpha \cdot \left(1 - \frac{\alpha}{2}\right) \cdot \left(\frac{1}{1 + \alpha} \right) \cdot x_{ij}^*, \]
% and for $\alpha = 3/4$, we have $s_{ij} \geq \nicefrac{15}{56} \cdot x_{ij}^*$. 
\end{proof}

\paragraph{Tiny inventory results.} Similar to the single-good setting, a direct implication of the proof of \Cref{thm:competitive-multi-good-policy} is that inventory capacity of $2$ or greater is sufficient to guarantee that \Cref{alg:main} (with the specifications mentioned) is a $\nicefrac{15}{56}$-competitive policy.

\section{Proof of \Cref{thm:1/2-competitive-single-good-policy} (impossibility result)}
\label{sec:impossibility}

In this section we prove that our single-good online policy's competitive ratio of $\nicefrac{1}{2}$ is optimal among all online policies. 
Our proof of this hardness result of \Cref{thm:1/2-competitive-single-good-policy} follows from what can be viewed as the adaptation of the oft-cited example for the tightness of the $\nicefrac{1}{2}$ competitive ratio for the classic prophet inequality problem.\footnote{The example consists of two buyers, the first with bid $1$, and the second with bid $1/\eps$ with probability $\eps$ and $0$ otherwise, for $\eps \rightarrow 0$. While the expected maximum is $2 - \eps$, no online strategy has expectation greater than $1$.}
 
\begin{proof}[Proof of \Cref{thm:1/2-competitive-single-good-policy} (impossibility result)]
We consider the following simple instance which we denote by $\mathcal{I}_\eps$. The supply rate of the seller's good is $\lambda = \eps$ and the perish rate is $\mu = 1$. There are two buyer types: the rare ``big spender'' who arrives with rate $\gamma_1 = \eps$ and bids $v_1= 1 + \nicefrac{1}{\eps}$, and the common ``miser'' who arrives with rate $\gamma_2 = \infty$ and bids $v_2 = 1$. One can trivially achieve expected average revenue $\eps$ for instance $\mathcal{I}_\eps$ online: simply sell each item to the common buyer which arrives immediately after it. On the other hand, for small $\eps$, an expected average revenue of roughly $2 \cdot \eps$ is achievable offline, as we now show.

\begin{cla}
\label{claim:expected-average-revenue-optimal-offline-tight-instance-competitive-ratio}
$\optoff(\calI_{\eps})\geq \eps + 1- \exp\left(-\nicefrac{\eps}{(1+\eps)}\right)$. 
%In particular, $\lim_{\eps\rightarrow 0^+}\optoff(\calI_\eps)\geq 2\eps$.
\end{cla}

\begin{proof}
We consider the offline policy which always permits a sale to a rare buyer (provided at least one item is available) and permits a sale to a common buyer only at the moment before an item perishes.
Note that this is indeed a valid offline policy.

By the PASTA property (\Cref{pasta}), letting $A$ denote the number of items available upon arrival of a buyer, the rate at which items of the good are sold to the rare buyer is $\eps \cdot \mathbb{P}\left[ A \geq 1 \right]$.
Since each item that is not sold to a rare buyer is sold to a common buyer who arrives just before the item perishes, the rate of selling to common buyers is the residual rate: $\eps - \eps \cdot \mathbb{P}\left[ A \geq 1 \right]$.
Therefore, the expected average revenue of this offline policy is
\[ \left(1 + \frac{1}{\eps} \right) \cdot \eps \cdot \mathbb{P}\left[ A \geq 1 \right] + 1 \cdot \left(\eps - \eps \cdot \mathbb{P}\left[ A \geq 1 \right] \right) = \eps + \mathbb{P}\left[ A \geq 1 \right]. \]
Lower bounding $\mathbb{P}\left[ A \geq 1 \right]$ as in \Cref{cor:Pr-available-lower-upper-bounds}, we have that the expected average revenue of this offline policy (and consequently of the optimal offline policy) is at least $\eps + 1 - \exp\left(- \nicefrac{\eps}{1+\eps} \right)$.
\end{proof}

Next, we show that the aforementioned trivial online policy which sells each item immediately to a common buyer and has expected average reward $\eps$, is optimal among all online policies.
\begin{cla}
\label{claim:expected-average-revenue-optimal-online-tight-instance-competitive-ratio}
$\opton(\calI_{\eps})\leq \eps.$
\end{cla}

\begin{proof}
By \Cref{cor:LP-off-geq-OPT-off-LP-on-geq-OPT-on}, $\opton(\calI_\eps) \leq \mathrm{LP}_\mathrm{on}(\calI_\eps)$. For instance $\mathcal{I}_\eps$, $\mathrm{LP}_\mathrm{on}$ is
\begin{align*}
    \max \quad & \left(1 + 1/\eps \right) \cdot x_1 + x_2 \\
    \textrm{s.t.}\quad & x_1 + x_2 \leq \eps \\
    & 1/\eps \cdot x_1 \leq \min\left\{ 1 - \exp(-\eps), \eps - (x_1 + x_2) \right\} \\
    & x_1, \ x_2 \geq 0.
\end{align*}
Fix some optimal solution $\{x_1^*, \ x_2^*\}$.
The value of this solution is
\[ \left( 1 + \frac{1}{\eps}\right) \cdot x^*_1 + x^*_2 \leq x_1^* + x^*_2 + \eps - (x^*_1 + x^*_2) = \eps. \]
Therefore, the expected average revenue of the optimal online policy is at most $\eps$.
\end{proof}

Combining \Cref{claim:expected-average-revenue-optimal-offline-tight-instance-competitive-ratio,claim:expected-average-revenue-optimal-online-tight-instance-competitive-ratio} and taking $\eps$ to zero, the theorem follows.
%\[ \lim_{\eps \rightarrow 0} \frac{\opton(\calI_\eps)}{\optoff(\calI_\eps)} \leq \lim_{\eps \rightarrow 0} \frac{\eps}{\eps + 1 - \exp\left(-\frac{\eps}{1 + \eps}\right)} = \frac{1}{2}. \qedhere \]
\end{proof}

\appendix
\section*{APPENDIX}

\section{Omitted proofs of \Cref{sec:prelims}}\label{appendix:prelims}

In this section, we provide proofs of lemmas deferred from \Cref{sec:prelims}, starting with results following from the queuing theory literature.

\PrAvailableExact*
\begin{proof}
The number of items of good $i$ available under an online policy which sells any available item to a buyer is captured by the CTMC in \Cref{fig:MC-birth-death-chain}.
This is a birth-death chain on state space $\{0, 1, \dots, C\}$ with transition rates $\alpha_q = \lambda_i$ from state $q-1$ to state $q$ and $\beta_q = q \cdot \mu_i + \gamma^*$ from state $q$ to state $q-1$, for all $1 \leq q \leq C$ as in \Cref{fig:MC-birth-death-chain}.

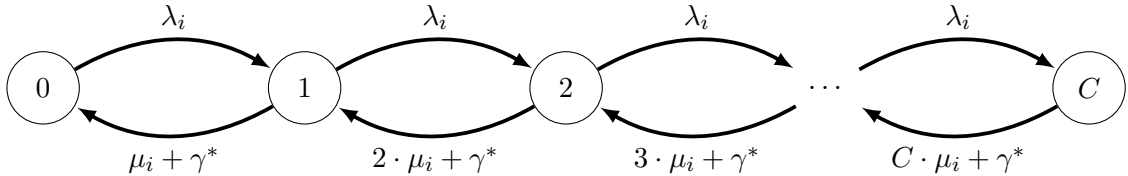
\begin{figure}[h]
\begin{center}
    \begin{tikzpicture}[font=\sffamily]
    % add the states
    \node[state,
        text=black,
        draw=black] (0) {$0$};
    \node[state,
        right=2.5cm of 0,
        text=black,
        draw=black] (1) {$1$};
    \node[state,
        right=2.5cm of 1,
        text=black,
        draw=black] (2) {$2$};
    \node[state,
        right=2.5cm of 2,
        text=black,
        draw=white] (3) {$\dots$};
    \node[state,
        right=2.5cm of 3,
        text=black,
        draw=black] (4) {$C$};
    
    % Connect the states with arrows
    \draw[every loop,
        auto=left,
        line width=0.5mm,
        >=latex,
        draw=black,
        fill=black]
        (0) edge[bend left, auto=left]  node {$\lambda_i$} (1)
        (1) edge[bend left, auto=left] node {$\mu_i + \gamma^*$} (0)
        (1) edge[bend left, auto=left]  node {$\lambda_i$} (2)
        (2) edge[bend left, auto=left] node {$2 \cdot \mu_i + \gamma^*$} (1)
        (2) edge[bend left, auto=left]  node {$\lambda_i$} (3)
        (3) edge[bend left, auto=left] node {$3 \cdot \mu_i + \gamma^*$} (2)
        (3) edge[bend left, auto=left]  node {$\lambda_i$} (4)
        (4) edge[bend left, auto=left] node {$C \cdot \mu_i + \gamma^*$} (3);
   \end{tikzpicture}
\end{center}
\vspace{-0.25cm}
\caption{CTMC of the number of items of good $i$ available under a policy selling available items to buyers which arrive at rate $\gamma^*$ when the seller has inventory capacity $C$.}
\label{fig:MC-birth-death-chain}
\end{figure}

From \cite[Section 3.1]{bolch2006queueing}, the stationary probability of having no items available (i.e. being at state 0) is
\[ \bbP_C \left[ A_i = 0 \right] = \left(1 + \sum_{q=1}^C \prod_{r=1}^q \frac{\alpha_r}{\beta_r}\right)^{-1} = \left(1 + \sum_{q=1}^C \prod_{r=1}^q \frac{\lambda_i}{r \cdot \mu_i + \gamma^*}\right)^{-1}. \qedhere \]
\end{proof}

\PrAvailableBounds*
\begin{proof}
The lower and upper bounds follow immediately from \Cref{lem:Pr-available-exact} and the fact that $r \cdot \mu_i \leq r \cdot \mu_i + \gamma^* \leq r \cdot (\mu_i + \gamma^*)$ for all $r \geq 1$.
\end{proof}

\section{Bounds implied by prior work}\label{appendix:priorwork}

In \cite{aouad2020dynamic}, Aouad and Sarita{\c{c}} present an algorithm for their dynamic matching problem, which they prove yields a $\frac{1}{4}(1-\nicefrac{1}{e})$-approximation of the optimal online algorithm for their problem. Essentially, their approach reduces their problem to the multi-good prophet inequality problem, for which they provide a $(1-\nicefrac{1}{e})$-approximation of the optimal online policy. We refer to \cite{aouad2020dynamic} for more details. In \Cref{sec:limitations-prior-LPs}, we show that this bound is inherent to their approach, since the LP benchmark they rely on cannot be used to prove a better than $1-\nicefrac{1}{e}$ approximation. 

\subsection{Competitive policies}

As stated before, the approach of \citet{collina2020dynamic}, which gives a $\nicefrac{1}{8}$-competitive algorithm for their problem, can be shown to have an improved bound for our problem, as the following lemma asserts.

\begin{lem}
\label{lem:1/3-competitive-single-good-WINE-analysis-WINE-LP}
Via the proof strategy of \citet{collina2020dynamic}, the algorithm from \cite{collina2020dynamic} for the single-good stationary prophet inequality problem can be shown to be a $\nicefrac{1}{3}$-competitive posted-price policy.
\end{lem}

\begin{proof}
Fix an instance $\calI$ of the single-good stationary prophet inequality problem.
Note that in the single-good problem, the algorithm from \cite{collina2020dynamic} reduces to \Cref{alg:main} with $\alpha=1$ and $\mathbf{x}^* = \arg \max_\mathbf{x} \mathrm{RB}_\mathrm{off}(\calI)$ and $w \triangleq \min\{1, \lambda\}$.
By a similar argument as in the proof of \Cref{obs:alg-posted-price}, this definition of $w$ guarantees that $p_j$ is a valid probability for all $j \in \calB$.

In order to lower bound the rate $s_j$ at which buyers of type $j$ are sold items of the good, the authors of \cite{collina2020dynamic} consider the following events:
\begin{itemize}
    \item $E^1$: An item of the good arrives.
    \item $E^2$: A buyer of any type arrives and the seller permits a sale.
    \item $E^3$: When there is exactly one item of the good available, this item perish. If there are more/fewer than one item available, $E_3$ follows an independent Poisson clock of rate $1$.
    \item $E_j^4$: A buyer of type j arrives and the seller permits a sale.
\end{itemize}
The authors of \cite{collina2020dynamic} note that a buyer of type $j$ is sold an item of the good if when event $E_j^4$ occurs, the most recent of events $E^1, E^2, E^3$ to have occurred prior to $E_j^4$ is $E^1$.
If this is indeed the case, that means that when the buyer of type $j$ arrives and the seller permits the sale, there is at least one item of the good available.
Notice that these four events are independent Poisson processes and therefore the rate at which buyers of type $j$ are sold items of the good is simply the rate at which event $E_j^4$ occurs times the probability that $E_1$ was the most recent of events $E^1, E^2, E^3$, which by the time reversibility of Poisson processes, is equal to the probability that $E_1$ occurs before either $E_2$ or $E_3$.
Therefore, we have
\[ s_j \geq \lambda_{E_j^4} \cdot \frac{\lambda_{E^1}}{\lambda_{E^1} + \lambda_{E^2} + \lambda_{E^3}}, \]
where we let $\lambda_E$ denote the rate at which event $E$ occurs.
Clearly, $\lambda_{E^1} = \lambda$, $\lambda_{E^3} = 1$, and for any buyer type $j \in \calB$, $\lambda_{E_j^4} = \gamma_j \cdot \alpha \cdot \frac{x_j^*}{\gamma_j \cdot w} = \frac{\alpha}{w} \cdot x_j^*$.
Also,
\[ \lambda_{E^2} = \sum_{j \in \calB} \gamma_j \cdot \alpha \cdot \frac{x_j^*}{\gamma_j \cdot w} = \alpha \sum_{j \in \calB} \frac{x_j^*}{w} \leq \alpha \cdot \frac{\lambda}{w}. \]
As a result,
\begin{equation}
\label{eq:s-ij-single-good-WINE-analysis-WINE-LP}
    s_j \geq \frac{\alpha}{w} \cdot \frac{\lambda}{\lambda + \alpha \cdot \frac{\lambda}{w} + 1} \cdot x_j^* = \alpha \cdot \frac{\lambda}{\min\{1, \lambda\} \cdot \left(\lambda + 1 \right) + \alpha \cdot \lambda} \cdot x_j^* \geq \frac{\alpha}{2 + \alpha} \cdot x_j^*,
\end{equation}
where the last inequality follows from the following claim.

\begin{cla}
For all $x \in \R_{> 0}$ and $\alpha \in [0, 1]$,
\begin{equation}
\label{eq:lower-bound-1/(2+alpha)}
    \frac{x}{\min\{1, x\} \cdot \left(x + 1 \right) + \alpha \cdot x} \geq \frac{1}{2 + \alpha}.
\end{equation}
\end{cla}

\begin{proof}
Let $f(x)$ be the left hand side of \Cref{eq:lower-bound-1/(2+alpha)}.
First suppose $\min\{1, x\} = 1$ (i.e., $x \geq 1$).
Under this assumption,
\[ f(x) = \frac{x}{x \cdot (1 + \alpha) + 1} = \frac{1}{1 + \alpha + \frac{1}{x}}, \]
which is increasing in $x$.
Therefore, the minimum is achieved at $x = 1$, in which case we have $f(1) = 1 / (2 + \alpha)$.
If instead $\min\{1, x\} = 1$ (i.e., $x \leq 1$),
\[ f(x) = \frac{x}{x \cdot \left(x + 1 \right) + \alpha \cdot x} = \frac{1}{x + 1 + \alpha} \]
which is decreasing in $x$.
Therefore, the minimum is again achieved at $x = 1$, completing the claim.
\end{proof}

As $\alpha / (2 + \alpha)$ is increasing in $\alpha$, this quantity is maximized for $\alpha = 1$, in which case we have $s_j \geq \nicefrac{1}{3} \cdot x_j^*$.
Since this holds for all buyer types $j \in \calB$, by linearity of expectation and \Cref{cor:RB-geq-OPT}, the expected average reward of the algorithm from \cite{collina2020dynamic} for $\calI$ is at least a third that of the optimal offline policy in expectation.
Furthermore, since $\alpha = 1$, an analogous proof of \Cref{obs:alg-posted-price} for $\mathbf{x}^* = \arg\max_\mathbf{x} \mathrm{RB}_\mathrm{off}(\calI)$ and $w \triangleq \min\{1, \lambda \}$ implies that this algorithm is a posted-price policy.
\end{proof}

A corollary of \Cref{lem:1/3-competitive-single-good-WINE-analysis-WINE-LP} is that when using the solution to $\mathrm{LP}_\mathrm{off}$ instead of $\mathrm{RB}_\mathrm{off}$ (and the corresponding definition of $w$), the techniques from \cite{collina2020dynamic} immediately yield a $0.435$ competitive ratio, although as we show in \Cref{sec:better-bounds}, this is still not the best possible among all online policies.

\begin{cor}
\label{cor:0.435-competitive-single-good-WINE-analysis-our-LP}
Via the proof strategy of \citet{collina2020dynamic}, the competitive ratio of the algorithm from \cite{collina2020dynamic} improves to $0.435$ when $\mathbf{x}^* = \mathrm{LP}_\mathrm{off}(\calI)$ and $w \triangleq 1 - \exp(-\lambda)$.
\end{cor}

\begin{proof}
Fix an instance $\calI$ of the single-good stationary prophet inequality problem.
Note that the algorithm from \cite{collina2020dynamic} with $\mathbf{x}^* = \arg\max_\mathbf{x} \mathrm{LP}_\mathrm{off}(\calI)$ and $w \triangleq 1 - \exp(-\lambda)$ is exactly \Cref{alg:main}, and we can lower bound the rate at which items of the good are sold to buyers of type $j \in \calB$ as in \Cref{lem:1/3-competitive-single-good-WINE-analysis-WINE-LP}.
From \Cref{eq:s-ij-single-good-WINE-analysis-WINE-LP} we have
\[ s_j = \alpha \cdot \frac{\lambda}{\left(1 - \exp(-\lambda)\right) \cdot (\lambda + 1) + \alpha \cdot \lambda} \cdot x_j^*. \]

In order to show that the right-hand expression above is at least $0.435 \cdot x_j^*$, we hold $\alpha$ fixed, let $f(\alpha, \lambda) = \alpha \cdot \lambda / \left( \left(1 - \exp(-\lambda)\right) \cdot (\lambda + 1) + \alpha \cdot \lambda \right)$, and consider the derivative of $f$ with respect to $\lambda$:
\[ \frac{\partial f(\alpha, \lambda)}{\partial \lambda} = \alpha \cdot \frac{e^\lambda \left(e^\lambda - 1 - \lambda - \lambda^2 \right)}{\left(1 + \lambda - e^\lambda \left(1 + \lambda + \alpha \cdot \lambda \right) \right)^2}. \]
The roots of $\partial f(\alpha, \lambda) / \partial \lambda$ are simply the roots of $e^\lambda - 1 - \lambda - \lambda^2$, one of which is trivially $\lambda_1^* = 0$.
In order to determine the other root(s), of which we show there is in fact only one, we simplify using the Taylor expansion of $e^\lambda$:
\[ e^\lambda - 1 - \lambda - \lambda^2 = \sum_{k=2}^\infty \frac{\lambda^k}{k!} - \lambda^2 = \lambda^2 \cdot \left( \sum_{k=2}^\infty \frac{\lambda^{k-2}}{k!} - 1 \right) = \lambda^2 \left( \sum_{k=3}^\infty \frac{\lambda^{k-2}}{k!} - \frac{1}{2} \right) = \lambda^2 \cdot \left( \sum_{k=1}^\infty \frac{\lambda^k}{(k+2)!} - \frac{1}{2} \right). \]
Clearly any other root $\lambda^* \neq 0$ satisfies
\[ \sum_{k=1}^\infty \frac{(\lambda^*)^k}{(k+2)!} = \frac{1}{2}, \]
and since the left-hand side is strictly increasing for $\lambda^* \geq 0$, there is exactly one value that satisfies this equality, which we denote by $\lambda_2^*$.
Numerically, we find $\lambda_2^* \approx 1.793$.

Therefore for all $\lambda > 0$, we have $f(\alpha, \lambda) \geq f(\alpha, \lambda_2^*)$.
The maximum of $f(\alpha, \lambda_2^*)$ for $\alpha \in [0, 1]$ is achieved at the right boundary, and we have $f(1, \lambda_2^*) \geq 0.435$.
Consequently, for all buyer types $j \in \calB$, $s_j \geq 0.435 \cdot x_j^*$, and the competitive ratio follows.
\end{proof}

Although at first glance they may seem quite different, the techniques from \cite{collina2020dynamic} detailed in \Cref{lem:1/3-competitive-single-good-WINE-analysis-WINE-LP} can actually be related to our proof of \Cref{thm:1/2-competitive-single-good-policy} by the stationary distribution of the CTMC that we analyze.
Indeed, the lower bound on the probability that an item is available in the single-good problem implied by the analysis from \cite{collina2020dynamic} simply captures the dynamics of this CTMC when the seller's inventory capacity is constrained to $1$.
From \Cref{lem:Pr-available-exact}, the probability that an item is available under the algorithm from \cite{collina2020dynamic} when $C = 1$ is at least
\[ 1 - \left( 1 + \frac{\lambda}{1 + \alpha \cdot \frac{\lambda}{w}} \right)^{-1} = 1 - \frac{1 + \alpha \cdot \frac{\lambda}{w}}{\lambda + 1 + \alpha \cdot \frac{\lambda}{w}} = \frac{\lambda}{\lambda + 1 + \alpha \cdot \frac{\lambda}{w}}. \]
This exactly matches the lower bound on the probability that event $E^1$ occurred most recently of events $E^1, E^2, E^3$ from \Cref{eq:s-ij-single-good-WINE-analysis-WINE-LP}, which \citet{collina2020dynamic} use to bound the probability that one or more items are available.
In contrast, by considering the dynamics of the CTMC with capacity of just $2$ or greater, our analysis enables us to improve the competitive ratio of the algorithm from \cite{collina2020dynamic} from $\nicefrac{1}{3}$ to $(1-\nicefrac{1}{(e-1)}) \approx 0.418$, as in the following lemma, and optimally to $\nicefrac{1}{2}$ in \Cref{thm:1/2-competitive-single-good-policy} when combined with our tighter benchmark.

\begin{lem}
\label{lem:WINE-alg-0.418-competitive}
The algorithm from \cite{collina2020dynamic} for the single-good stationary prophet inequality problem is a $(1-\nicefrac{1}{(e-1)})$-competitive posted-price policy.
\end{lem}

\begin{proof}
Fix an instance $\calI$ of the single-good stationary prophet inequality problem.
Recall from \Cref{lem:1/3-competitive-single-good-WINE-analysis-WINE-LP} that in the single-good stationary prophet inequality problem, the algorithm from \cite{collina2020dynamic} reduces to \Cref{alg:main} with $\mathbf{x^*} = \arg\max_\mathbf{x} \mathrm{RB}_\mathrm{off}(\calI)$ and $w \triangleq \min\{1, \lambda\}$.

By Constraint~\ref{eqn:flow-constraint-seller}, the arrival rate of buyers to whom the seller permits a sale is $\gamma^* \leq \lambda / w = \lambda / \min\{1, \lambda\}$, and therefore by \Cref{cor:Pr-available-lower-upper-bounds} we have the following, for any buyer type $j \in \calB$:
\begin{equation}
\label{eqn:Pr-available-divided-by-m*-geq-e-2-divided-by-e-1}
    s_j = \frac{\bbP\left[ A \geq 1 \right]}{w} \cdot x_j^* \geq \frac{1 - \left(1 + \sum_{q=1}^\infty \prod_{r=1}^q \frac{\lambda}{r + \frac{\lambda}{\min\{1, \lambda\}}} \right)^{-1}}{\min\{1, \lambda\}} \cdot x_j^*
\end{equation}

\begin{cla}
\label{cla:Pr-avail-divided-by-min-lower-bound}
For all $x \in \mathbb{R}_{> 0}$,
\[ \frac{1 - \left(1 + \sum_{q=1}^\infty \prod_{r=1}^q \frac{x}{r + \frac{x}{\min\{1, x\}}} \right)^{-1}}{\min\{1, x\}} \geq 1 - \frac{1}{e-1}. \]
\end{cla}

\begin{proof}
We consider two cases. First suppose $\min\{1, x\} = 1$.
The claim reduces to proving that for all $x \geq 1$,
\[ 1 + \sum_{q=1}^\infty \prod_{r=1}^q \frac{x}{r+x} \geq e - 1. \]
By a simple inductive proof, the left-hand side of the above equation is increasing in $x$ and therefore for $x \geq 1$,
\[ 1 + \sum_{q=1}^\infty \prod_{r=1}^q \frac{x}{r+x} \geq 1 + \sum_{q=1}^\infty \prod_{r=1}^q \frac{1}{r+1} = \sum_{q=1}^\infty \frac{1}{q!} = e - 1. \]

Otherwise, if $\min\{1, x\} = x$, we must show that for all $x \in [0, 1]$
\[ \frac{1 - \left(1 + \sum_{q=1}^\infty \prod_{r=1}^q \frac{x}{r + 1} \right)^{-1}}{x} \geq 1 - \frac{1}{e-1}. \]
Since $1 + \sum_{q=1}^\infty \prod_{r=1}^q \frac{x}{r + 1} = 1 + \frac{1}{x} \cdot (e^x - 1 - x) = \frac{1}{x} \cdot (e^x - 1)$, this is equivalent to proving that
\[ \frac{1}{x} - \frac{1}{e^x - 1} \geq 1 - \frac{1}{e - 1} \]
for all $x \in [0, 1]$.
We denote the left-hand side of the expression above by $f(x)$ and take its derivative:
\[ f^\prime(x) = \frac{e^x}{\left(e^x - 1 \right)^2} - \frac{1}{x^2}. \]
We show that $f^\prime$ is negative for all $x \in [0, 1]$ and therefore the minimum of $f$ on the interval $[0, 1]$ is achieved at the left boundary.
To this end, it suffices to show that $x \cdot e^{x/2} \leq e^x - 1$ for all $x \in [0, 1]$.
By the Taylor expansion of $e^x$,
\[ x \cdot e^{x/2} = \sum_{k=0}^\infty \frac{x^{k+1}}{2^k \cdot k!} \leq \sum_{k=0}^\infty \frac{x^{k+1}}{(k+1)!} = \sum_{k=1}^\infty \frac{x^k}{k!} = e^x - 1, \]
where the inequality holds since $2^k \geq k + 1$ for all $k \geq 1$.
Therefore, $f(x) \geq f(1) = 1 - 1/(e-1)$ for all $x \in [0,1]$.
\end{proof}

\Cref{eqn:Pr-available-divided-by-m*-geq-e-2-divided-by-e-1,cla:Pr-avail-divided-by-min-lower-bound} together yield $s_j \geq (1 - \nicefrac{1}{(e-1)}) \cdot x_j^*$ for all buyer types $j \in \calB$.
The lemma follows from linearity of expectation and \Cref{cor:RB-geq-OPT}.
\end{proof}

\section{Omitted proofs of \Cref{sec:tighter-lps}}\label{appendix:tighter-lps}

In this section we provide proofs on the gap between the reward benchmarks $\textrm{RB}_{\textrm{off}}$ and $\textrm{RB}_{\textrm{on}}$ compared to $\optoff$ and $\opton$. That is, we prove the following restated observations.

\RBoffgap*
\RBongap*

To prove the above observations, we consider instances of the stationary prophet inequality problem with a single good and a single buyer type.
The convenience of such an instance  $\calI$ is that the optimal policies are trivial: the seller sells an item to any buyer that arrives, provided at least one item is available, with the optimal offline policy selling to the earliest departing item.

\begin{proof}[Proof of \Cref{obs:RBoffgap}]
Consider the following instance, which we denote by $\calI_1$: the seller's good is supplied and perishes at rate $\lambda = \mu = 1$, and there is a single buyer type which arrives with rate $\gamma = 1$ and bids $v = 1$. 
By the PASTA property \cite{wolff1982poisson}, the rate at which items are sold is simply the rate at which buyers arrive and observe an available item.
From the upper bound of \Cref{cor:Pr-available-lower-upper-bounds}, we have
\[ \bbP[ A \geq 1 ] \leq 1 - \left(1 + \sum_{q=1}^\infty \frac{1}{(q+1)!} \right)^{-1} = 1 - \frac{1}{e - 1},
\]
where $A$ denotes the number of available items.
We conclude that $\optoff(\calI_1)\leq v\cdot \gamma\cdot \bbP[ A \geq 1 ] = 1 - \nicefrac{1}{(e-1)}.$
%As previously noted, the optimal online policy is to simply sell an item to a buyer whenever possible. By \Cref{lem:Pr-available-exact}, we have
%\[ \bbP\left[ A \geq 1 \right] = 1 - \left(1 + \sum_{q=1}^\infty \frac{1}{(q+1)!} \right)^{-1} = 1 - \frac{1}{e - 1} = \frac{e-2}{e-1}. \]
%It follows from \Cref{eqn:tight-instance-OPT-on-vs-LP} that $\opton\left(\calI_1\right) = \nicefrac{(e-2)}{(e-1)}$.
On the other hand, we have that $\mathrm{RB}_\mathrm{off}(\calI_1)= \max \{x \mid  x \leq 1,\, x \geq 0\} = 1$.
\end{proof}

\begin{proof}[Proof of \Cref{obs:RBongap}]
Consider the following instance, which we denote by $\calI_\lambda$: the seller's good is supplied at rate $\lambda \geq 2$ and perishes at rate $\mu = \lambda - 1$, and there is a single buyer type which arrives with rate $\gamma = 1$ and bids $v = 1$.
By \Cref{pasta,lem:Pr-available-exact}, the expected average revenue of this optimal online policy, which simply sells an item to a buyer whenever possible, is the rate at which buyers arrives, times the stationary probability that there is at least one item available:
% The optimal online policy is to sell an item to a buyer whenever possible. 
% As such, by \Cref{pasta} and \Cref{lem:Pr-available-exact}, the expected average revenue of this policy is simply the bid of the buyer times the rate at which buyers arrive, times the stationary probability that there is at least one item available:
\[ \opton\left(\calI\right) = v \cdot \gamma \cdot \bbP\left[ A \geq 1 \right] = \bbP\left[ A \geq 1 \right] = 1 - \left( 1 + \sum_{q=1}^\infty \prod_{r=1}^q \frac{\lambda}{r \cdot (\lambda - 1) + 1} \right)^{-1}. \]
Therefore,
we have that $\lim_{\lambda \rightarrow \infty} \opton\left(\calI_\lambda\right) = 1 - \nicefrac{1}{e}.$
On the other hand, we have that $\mathrm{RB}_\mathrm{on}(\calI_\lambda) =    \max \{x\geq 0 \mid x \leq \min\{1, \lambda,   \lambda / (\lambda - 1)\}\}$, where we note that \Cref{ec20-constraint} and \Cref{eqn:wine-20-constraint} both simplify to $x\leq \lambda / (\lambda - 1)$ for the instance $\calI_{\lambda}$. Therefore, for any value of $\lambda \geq 1$, we have that $\mathrm{RB}_\mathrm{on}\left(\calI_\lambda\right) \geq 1$.
Taking $\lambda\rightarrow \infty$, the observation follows.
\end{proof}

\section{Omitted proofs of \Cref{sec:better-bounds}}\label{appendix:better-bounds}
In this section we provide proofs of claims and lemmas deferred from \Cref{sec:better-bounds}, starting with proofs concerning approximability of the single-good problem.

\subsection{Single-good problem}
\subsubsection{Approximating $\optoff$}

\LowerBoundOneHalf*
\begin{proof}
Denote the LHS of this inequality by $f(x)$. 
% f[x_] := (1 - (1 + x/(1 + x/(1 - Exp[-x]))*(1 +           x/(2 + x/(1 - Exp[-x]))))^(-1))/(1 - Exp[-x])
We note that $\lim_{x\rightarrow 0^+} f(0)=\frac{1}{2}$,
% Limit[f[x], {x -> 0}]
and so we would like to prove that $f(x)$ is monotone increasing in $x$. Taking the derivative of the LHS and simplifying it, we get the following (rather unwieldy) derivative.
% Simplify[f'[x]]
% Out = (E^x (-4 - 8 x - 6 x^2 - 4 x^3 - x^4 + 4 E^(3 x) (1 + x)^2 - E^(2 x) (12 + 24 x + 17 x^2 + 9 x^3 + 3 x^4) + E^x (12 + 24 x + 19 x^2 + 12 x^3 + 4 x^4)))/(2 + 2 x + x^2 + E^(2 x) (2 + 5 x + 3 x^2) - E^x (4 + 7 x + 3 x^2))^2
$$f'(x):=\frac{\sum_{i=0}^3 g_i(x)}{e^{-x}\cdot (2 + 2 x + x^2 +
   e^{2 x} \cdot (2 + 5 x + 3 x^2) - e^x \cdot(4 + 7 x + 3 x^2))^2},$$
where, grouping by powers of $e^x$, these $g_i(x)$ are defined as follows.
   $$g_i(x)=\begin{cases}
       e^{0} \cdot (-4 - 8 x - 6 x^2 - 4 x^3 - x^4)  & i = 0 \\
       e^x \cdot (12 + 24 x + 19 x^2 + 12 x^3 + 4 x^4) & i = 1\\
       e^{2 x} \cdot (-12 - 24 x - 17 x^2 - 9 x^3 - 3 x^4)  & i = 2\\
       e^{3 x} \cdot (4 +8x + 4x^2) & i = 3.
   \end{cases}
    $$
Now, the denominator of the above form for $f'(x)$ is easily seen to be positive for all $x\geq 0$ (indeed, it is positive for all $x\in \mathbb{R}$, since it is the product of an exponential and a square). Therefore, to prove that $f'(x)\geq 0$ for all $x\geq 0$, we need only prove that $g(x):=\sum_{i=0}^3 g_i(x) \geq 0$ for all $x\geq 0$. 
Now, $g(x)$ is the sum of products of analytic functions which are in particular equal to their Taylor expansions around zero. Therefore, $g(x)$ is equal to its Taylor expansion around zero, and we can write it as
$$g(x) = \sum_{n=0}^{\infty} a_n\cdot x^n,$$
where, using the coefficients of the Taylor expansion of $e^{k\cdot x} = \sum_{n=0}^\infty \frac{k^n}{n!}\cdot x^{n},$ we have that
$$a_n =
    b_n - 4\cdot \mathds{1}[n=0] - 8\cdot \mathds{1}[n=1] - 6\cdot \mathds{1}[n=2] - 4\cdot \mathds{1}[n=3] - 1\cdot \mathds{1}[n=4],$$
where 
\begin{align*}
    b_n & = \frac{12-12\cdot 2^n + 4\cdot 3^n}{n!} + \frac{24-24\cdot 2^{n-1} + 8\cdot 3^{n-1}}{(n-1)!}
    \\
    & + \frac{19-17\cdot 2^{n-2} + 4\cdot 3^{n-2}}{(n-2)!}
    + \frac{12-9\cdot 2^{n-3}}{(n-3)!}
    + \frac{4-3\cdot 2^{n-4}}{(n-4)!}.
\end{align*}
The above $a_n$ are all non-negative. This can be proven by inspection for $n\leq 45$, while for $n\geq 45$, we have that 
$a_n \geq \frac{4\cdot 3^n}{n!} - \frac{(12+24+17+9+3)\cdot 2^n}{(n-4)!} \geq \frac{4\cdot 3^n}{n!} - \frac{65 \cdot 2^n}{n!/n^4} \geq 0,$
where the last inequality is equivalent to  $(3/2)^n\geq \frac{65}{4}\cdot n^4$, which holds for all $n\geq 45$.
We conclude that, since all $a_n$ are non-negative, we have that $g(x) = \sum_{n=0}^\infty a_n \cdot x^n \geq 0$ for all $x\geq 0$. Recalling that this implies that $f'(x)\geq 0$, we find that, indeed, $f(x)\geq \frac{1}{2}$ for all $x\geq 0$, as claimed.
\end{proof}

\subsubsection{Approximating $\opton$}

In our proof of \Cref{thm:0.656-approximate-single-good-policy}, we defined two auxiliary functions, $g_1(C, x) \triangleq (1 - \left(\sum_{q=0}^C \frac{x^q}{q!} \right)^{-1})/x$ and 
$g_2(C,x)\triangleq \frac{1 - \left(\frac{1}{6} + \frac{5}{6} \cdot \sum_{q=0}^C \frac{1}{q!} \left( \frac{6 \cdot w}{5} \right)^q \right)^{-1}}{w}$, which we claimed are monotone decreasing, as we now prove.

\begin{fact}\label{monotonicity-of-g1}
    For any fixed $C$, the function $g_1(C,w)$ is monotone decreasing in $w$.
\end{fact}

%\gOneLowerBound*
\begin{proof}
%We show that for fixed $C$, $g_1(C, w)$ is decreasing in $w$, and therefore the minimum of $g_1(C, w)$ on the interval $w \in [0, 1 - \exp(-12/5)]$ is achieved at the right boundary.
Taking the derivative of $g_1(C, w)$ with respect to $w$ yields
\[ \frac{\partial g_1(C, w)}{\partial w} = \frac{1}{w} \cdot \left(\frac{\sum_{q=0}^{C-1} \frac{w^q}{q!}}{\left( \sum_{q=0}^C \frac{w^q}{q!} \right)^2} - g_1(C, w) \right), \]
and it therefore suffices to show that $g_1(C, w) \geq \left( \sum_{q=0}^{C-1} \frac{w^q}{q!} \right) / \left( \sum_{q=0}^{C} \frac{w^q}{q!} \right)^2$.
Rearranging terms and using the definition of $g_1(C, w)$, we find that this is equivalent to the following inequality, which indeed holds for all $w\geq 0$,
\begin{align*}1 + w \cdot \frac{\sum_{q=0}^{C-1} \frac{w^q}{q!}}{\sum_{q=0}^C \frac{w^q}{q!}} & \geq 1. \qedhere
\end{align*}
\end{proof}

\begin{fact}\label{monotonicity-of-g2}
    For any fixed $C$, the function $g_2(C,w)$ is monotone decreasing in $w$.
\end{fact}
%\gTwoLowerBound*
\begin{proof}
%We show that for fixed $C$, $g_2(C, w)$ is decreasing in $w$, and therefore the minimum of $g_2(C, w)$ on the interval $w \in [1 - \exp\left(-\nicefrac{12}{5} \right), 1]$ is achieved at the right boundary.
Taking the derivative of $g_2(C, w)$ with respect to $w$ yields
\[ \frac{\partial g_2(C, w)}{\partial w} = \frac{1}{w} \cdot \left(\frac{\sum_{q=0}^{C-1} \frac{1}{q!} \left(\frac{6 \cdot w}{5} \right)^q}{\left( \frac{1}{6} + \frac{5}{6} \cdot \sum_{q=0}^C \frac{1}{q!} \left( \frac{6 \cdot w}{5} \right)^q \right)^2} - g_2(C, w) \right), \]
and it therefore suffices to show that $g_2(C, w) \geq \left( \sum_{q=0}^{C-1} \frac{1}{q!} \left(\frac{6 \cdot w}{5} \right)^q \right) / \left( \frac{1}{6} + \frac{5}{6} \cdot \sum_{q=0}^C \frac{1}{q!} \left( \frac{6 \cdot w}{5} \right)^q \right)^2$.
Rearranging terms using the definition of $g_2(C, w)$, we find that this is equivalent to showing the following inequality, which indeed holds for all $w\geq 0$.
\begin{align*} 
1 + w \cdot \frac{\sum_{q=0}^{C-1} \frac{1}{q!} \left(\frac{6 \cdot w}{5} \right)^q}{\frac{1}{6} + \frac{5}{6} \cdot \sum_{q=0}^C \frac{1}{q!} \left( \frac{6 \cdot w}{5} \right)^q} & \geq 1. \qedhere \end{align*}
%which is indeed true for $w \geq 0$.
\end{proof}

Finally, we prove that limiting the inventory sufficiently results in a worse approximation of the optimal (unbounded-inventory) online algorithm.
\UpperBoundApproxRatioInventoryC*
\begin{proof}
We consider the following instance of the single-good stationary prophet inequality problem: the seller's good is supplied at rate $\lambda > 0$ and perishes at rate $\mu = 1$, and there is a single buyer type which arrives at rate $\gamma = \lambda$ and bids $v = 1$.
As there is only one buyer type, the optimal online policy, regardless of the seller's inventory capacity, is to always sell an available item to any buyer that arrives.
Therefore, by the PASTA property \cite{wolff1982poisson} and \Cref{lem:Pr-available-exact}, the expected average revenue of this policy when the seller has inventory capacity $C \in \mathbb{Z}_{> 0}$ is
\[ v \cdot \gamma \cdot \mathbb{P}_C \left[ A \geq 1 \right] = \lambda \cdot \left(1 - \left(1 + \sum_{q=1}^C \prod_{r=1}^q \frac{\lambda}{r + \lambda} \right)^{-1} \right). \]
As $\lambda$ goes to infinity, the expected average revenue of the optimal online policy when the seller has inventory capacity $C$ relative to that of the optimal online policy when the seller has unbounded inventory capacity goes to
\[ \lim_{\lambda \rightarrow \infty} \frac{\lambda \cdot \mathbb{P}_C \left[ A \geq 1 \right]}{\lambda \cdot \mathbb{P} \left[ A \geq 1 \right]} = \lim_{\lambda \rightarrow \infty} \frac{1 - \left(1 + \sum_{q=1}^C \prod_{r=1}^q \frac{\lambda}{r + \lambda} \right)^{-1}}{1 - \left(1 + \sum_{q=1}^\infty \prod_{r=1}^q  \frac{\lambda}{r + \lambda} \right)^{-1}} = \frac{1 - \left(1 + \sum_{q=1}^C 1 \right)^{-1}}{1 - \left(1 + \sum_{q=1}^\infty 1 \right)^{-1}} = \frac{C}{C+1}. \qedhere \]
\end{proof}

\subsection{Multi-good problem}\label{appendix:multi-good}

\PrReachesAndAvailable*
\begin{proof}
Fix a buyer type $j \in \calB$ and an ordering $\sigma$ of the goods.
Note that the execution of \Cref{alg:main} when a buyer of type $j$ arrives is equivalent to the following: Before the buyer considers any of the goods, the seller determines which goods to permit a sale of, meaning the seller samples a set of permissible goods from the product distribution $Ber(p_{1j}) \times \dots \times Ber(p_{nj})$.
Then, as the buyer iterates through the goods according to order $\sigma$, the seller sells this buyer an item of the first good that he reaches that is both available and permissible for this buyer.
Therefore, in order for the buyer to reach good $i$, there must be no items of good $i^\prime$ available for each good $i^\prime$ that precedes $i$ in the order $\sigma$ (i.e, $\sigma(i^\prime) < \sigma(i)$) and which is permissible for the buyer.

Fix the set $H_j \subseteq \calG$ of permissible goods for this buyer, and let $D_{i, \sigma, H_j} = \{ i ^\prime \in H_j : \sigma(i^\prime) < \sigma(i) \}$.
Since every available item is also present, having no items of good $i^\prime$ present for each $i^\prime \in D_{i, \sigma, H_j}$ is a sufficient condition to guarantee that the buyer reaches good $i$.
Therefore, letting $Y \in \R^{2n}$ be the vector whose elements, which we refer to as $Y_{A_i}$ and $Y_{P_i}$, represent the number of items of good $i$ available and the negative of the number of items of good $i$ present, respectively, under \Cref{alg:main}, we have
\begin{equation}
\label{eq:R-ij-A-i-geq-Y}
    \bbP_C \left[ R_{ij} \wedge A_i \geq 1 \mid \sigma, H_j \right] \geq \bbP_C \left[ Y \geq e_{A_i} - \infty \cdot \sum_{i^\prime \in \bar D_{i, \sigma, H_j}} e_{P_{i^\prime}} \,\middle|\, \sigma,\, H_j \right].
\end{equation}
We use $e_{A_i}$ and $e_{P_i}$ to denote the vectors with all zeros except at the elements corresponding to $A_i$ and $P_{i}$ which are $1$.
We can think of $Y$ as simply the (augmented) state of the marketplace under \Cref{alg:main}, where the set $\calY_C \subseteq \mathbb{R}^{2n}$ of valid states is such that for any $y \in \calY_C$, we have $0 \leq y_{A_i} \leq C$, $y_{P_i} \leq 0$, and $y_{A_i} \leq |y_{P_i}|$.
Under \Cref{alg:main}, the stochastic process governing $Y$ is described by intensity matrix $Q_C$, where for any $y, y^\prime \in \calY_C$,
\[
    Q_C(y, y^\prime) = \begin{cases}
        \lambda_i &  y^\prime = y + e_{A_i} - e_{P_i} \text{ and } y_{A_i} < C  \\
        y_{A_i} \cdot \mu_i &  y^\prime = y - e_{A_i} + e_{P_i} \\
        \left(|y_{P_i}| - y_{A_i} \right) \cdot \mu_i &  y^\prime = y + e_{P_i} \\
        \alpha \cdot \sum_{j \in \calB} \bbP_C \left[ R_{ij} \mid y \right]\cdot \frac{x_{ij}^*}{w_i} &  y^\prime = y - e_{A_i} \text{ and } y_{A_i} > 0 \\
        0 & \text{o.w.}
        \end{cases}
\]
and $Q_C(y, y) = - \sum_{y^\prime \in \calY_C : y^\prime \neq y} Q_C(y, y^\prime)$.

Although the availability of good $i$ and the presence of other goods $i^\prime \neq i$ are correlated under $Q_C$, we show that $Y$ stochastically dominates a stochastic process $\tilde Y$ under which they are in fact independent.
Informally, we define the dynamics governing $\tilde Y$ to correspond to, in some sense, every good (simultaneously) coming first in the ordering such that an arriving buyer reaches each good with probability $1$.
Put another way, $\tilde Y$ can be thought of as a collection of $n$ independent single-good instances, where each instance consists of a different good $i \in \calG$ and the full set of buyers $\calB$.
Of course, this does not reflect the dynamics of \Cref{alg:main} (or any other of feasible policy for the multi-good problem) but is still a useful tool, as we will see.
%(This corresponds to each good $i$ in some sense always being first in the ordering, achieved by each good having an independent Poisson process with rate $\gamma_j$ independently of all other processes.)
More specifically, we let $\tilde Y$ represent the state of a stochastic process on the same space $\calY_C$ governed by intensity matrix $\tilde Q_C$, where for any $y, y^\prime \in \calY_C$,
\[
    \tilde Q_C(y, y^\prime) = \begin{cases}
    \lambda_i &  y^\prime = y + e_{A_i} - e_{P_i} \text{ and } y_{A_i} < C \\
    y_{A_i} \cdot \mu_i &  y^\prime = y - e_{A_i} + e_{P_i} \\
    \left(|y_{P_i}| - y_{A_i} \right) \cdot \mu_i &  y^\prime = y + e_{P_i} \\
    \alpha \cdot \frac{\lambda_i}{w_i} &  y^\prime = y - e_{A_i} \text{ and } y_{A_i} > 0 \\
    0 & \text{o.w.}
    \end{cases}
\]
and $\tilde Q_C(y, y) = - \sum_{y^\prime \in \calY_C : y^\prime \neq y} \tilde Q_C(y, y^\prime)$.
Observe that $Q_C$ and $\tilde Q_C$ are identical except for the rate at which they transition to states with exactly one fewer available item, a rate which is higher for $\tilde Y$.
Intuitively, this creates more ``downwards pressure'' for variable $\tilde Y$ than for $Y$. This intuition is borne out by the following lemma, 
which we prove following this proof.

\begin{restatable}{cla}{YDominantesYTilde}
\label{cla:Y-dominates-Y-tilde}
For any permutation $\sigma$ over $\calG$ and any $H_j \subseteq \calG$, we have that
$$[Y \mid \sigma,\, H_j] \succeq [\tilde Y \mid \sigma,\, H_j].$$
\end{restatable}

\Cref{cla:Y-dominates-Y-tilde} implies that for any $y \in \calY$, $\bbP_C \left[ Y \geq y \right] \geq \bbP_C \left[ \tilde Y \geq y \right]$ and therefore
\begin{equation}
\label{eq:Y-geq-tilde-Y}
    \bbP_C \left[ Y \geq e_{A_i} - \infty \cdot \sum_{i^\prime \in \bar D_{i, \sigma, H_j}} e_{P_{i^\prime}} \,\middle|\, \sigma, H_j \right] \geq \bbP_C \left[ \tilde Y \geq e_{A_i} - \infty \cdot \sum_{i^\prime \in \bar D_{i, \sigma, H_j}} e_{P_{i^\prime}} \,\middle|\, \sigma, H_j \right].
\end{equation}
Due to the independence of the availability of good $i$ and the presence of any other good $i^\prime \neq i$ under $\tilde Q_C$, we have
\begin{equation}
\label{eq:tilde-Y-geq-tilde-A-i-R-ij}
    \bbP_C \left[ \tilde Y \geq e_{A_i} - \infty \cdot \sum_{i^\prime \in \bar D_{i, \sigma, H_j}} e_{P_{i^\prime}} \,\middle|\, \sigma, H_j \right] = \bbP_C \left[ \tilde Y_{A_i} \geq 1 \,\middle|\, \sigma, H_j \right] \cdot \bbP_C \left[ \tilde Y_{P_{i^\prime}} = 0 \ \forall\, i^\prime \in D_{i, \sigma, H_j} \,\middle|\, \sigma, H_j \right].
\end{equation}
Observe that the event $\tilde Y_{A_i} \geq 1$ is exactly the event $\tilde A_i \geq 1$ and furthermore, that the availability of good $i$ under $\tilde Q$ does not depend on either the ordering $\sigma$ or the set of permissible goods $i^\prime \in H_j$.
Also, conditional on $\sigma$ and $H_j$, the event $\tilde Y_{P_{i^\prime}} = 0$ for all $i^\prime \in D_{i, \sigma, H_j}$ corresponds to the event that no good $i^\prime$ which is permissible for the buyer precedes $i$ in the ordering and is present, which is exactly $\tilde R_{ij}$.
Combining \Cref{eq:R-ij-A-i-geq-Y,eq:Y-geq-tilde-Y,eq:tilde-Y-geq-tilde-A-i-R-ij}, 
\Cref{eq:R-ij-A-i-geq-tilde-R-ij-tilde-A-i} follows.
\end{proof}

\begin{proof}[Proof of \Cref{cla:Y-dominates-Y-tilde}]
By \Cref{lem:stochastic-dominance}, it suffices to show that for all $\tilde y, y \in \calY_C$ and every upward closed set $S \subseteq \calY_C$,
\[ \sum_{\Delta : \tilde y + \Delta \in S} \tilde Q_C(\tilde y, z) \leq \sum_{\Delta : \tilde y + \Delta \in S} Q_C(y, z) \]
if $\tilde y \leq y$ and either $\tilde y, y \notin S$ or $\tilde y, y \in S$.
We fix $\tilde y, y$ and consider both the two cases individually.

\paragraph{Case 1:}
Suppose $\tilde y, y \notin S$.
For every $\Delta$ such that $\tilde y + \Delta \in S$, and $y+\Delta\in \calY_C$, we have $y + \Delta \in S$ since $\tilde y \leq y$ and $S$ is upward closed.
The upward closedness of $S$ also implies that since $\tilde y \notin S$, there must exist some coordinate $k$ such that $\Delta_k > 0$.
Inspecting the intensity matrix $\tilde Q$, it is clear that there are exactly three forms that $\Delta$ can take such that $\tilde Q_C(\tilde y, \tilde y + \Delta) > 0$: (1) $\Delta_i^{(1)} = e_{A_i} - e_{P_i}$, (2) $\Delta_i^{(2)} = e_{P_i} - e_{A_i} $, or (3) $\Delta_i^{(1)}= e_{P_i}$, for some $i \in \calG$.
Note that for (1), $\tilde y \leq y$, $\tilde y + \Delta \in S$, and $S \subseteq \calY_C$ together imply that $\tilde y_{A_i} = y_{A_i} < C$; otherwise, $\tilde y + e_{A_i} \leq y$ and therefore $\tilde y + \Delta \leq y$, meaning $y \in S$ by the upward closedness of $S$, contradicting $y \notin S$. We conclude that if $\tilde y + \Delta \in S$, then $y+ \Delta\in \calY_C$, and hence $y_\Delta \in S$.
Similarly for (2) and (3), $\tilde y \leq y$ and $\tilde y + \Delta \in S$ imply that $\tilde y_{P_i} = y_{P_i}$; otherwise, if $\tilde y_{P_i} < y_{P_i}$, then $\tilde y + e_{P_i} \leq y$ and therefore $\tilde y + \Delta \leq \tilde y + e_{P_i} \leq y$, which, by the upward closedness of $S$, contradicts $y \notin S$.
We now relate the intensities under $\tilde{Q}_C$ and $Q_C$.

Fix $\Delta$ such that $y + \Delta \in S$ and first suppose that $\Delta = \Delta_i^{(1)}$ for some $i \in \calG$.
Comparing the intensity matrices, we have that
\[ \tilde Q_C\left(\tilde y, \tilde y + \Delta_i^{(1)}\right) = \lambda_i = Q_C\left(y, y + \Delta_i^{(1)}\right). \]
Next suppose that $\Delta = \Delta_i^{(2)}$ for some $i \in \calG$.
Since $\tilde y + \Delta_i^{(2)} \in S$, by the upward closedness of $S$, $\tilde y + \Delta_i^{(3)} \in S$ and also $y + \Delta_i^{(2)}, y + \Delta_i^{(3)} \in S$.
Furthermore, we have that
\[ \tilde Q_C\left(\tilde y, \tilde y + \Delta_i^{(2)}\right) + \tilde Q_C\left(\tilde y, \tilde y + \Delta_i^{(3)}\right) = |\tilde y_{P_i}| \cdot \mu_i = |y_{P_i}| \cdot \mu_i = Q_C\left(y, y + \Delta_i^{(2)}\right) + Q_C\left(y, y + \Delta_i^{(3)}\right), \]
where the equality is due to the fact that $\tilde y_{P_i} = y_{P_i}$, as argued in the preceding paragraph.
Lastly, suppose that $\Delta = \Delta_i^{(3)}$ for some $i \in \calG$ such that $\tilde y + \Delta_i^{(2)} \notin S$.
If $\tilde y_{A_i} = y_{A_i}$, then
\[ \tilde Q_C\left(\tilde y, \tilde y + \Delta_i^{(3)}\right) = \left( |\tilde y_{P_i}| - \tilde y_{A_i} \right) \cdot \mu_i = \left( |y_{P_i}| - y_{A_i} \right) \cdot \mu_i = Q_C\left(y, y + \Delta_i^{(3)}\right). \]
If $\tilde y_{A_i} < y_{A_i}$, then $\tilde y \leq y - e_{A_i}$, which implies $\tilde y + \Delta_i^{(3)} \leq y + \Delta_i^{(2)}$, from which it follows that $y + \Delta_i^{(2)} \in S$ by the upward closedness of $S$.
Therefore,
\[ \tilde Q_C\left(\tilde y, \tilde y + \Delta_i^{(3)}\right) = \left(|\tilde y_{P_i}| - \tilde y_{A_i} \right) \cdot \mu_i \leq |\tilde y_{P_i}| \cdot \mu_i = |y_{P_i}| \cdot \mu_i = Q_C\left(y, y + \Delta_i^{(2)}\right) + Q_C\left(y, y + \Delta_i^{(3)} \right). \]

Putting this all together, it follows that
\[ \sum_{\Delta : \tilde y + \Delta \in S} \tilde Q_C(\tilde y, \tilde y + \Delta) \leq \sum_{\Delta : y + \Delta \in S} Q_C(y, y + \Delta). \]

\paragraph{Case 2:}
Suppose $y, \tilde y \in S$.
In this case it suffices to show that
\[ \sum_{\Delta : \tilde y + \Delta \in \calY_C \setminus S} Q_C(y, z) \leq \sum_{\Delta : \tilde y + \Delta \in \calY_C \setminus S} \tilde Q_C(\tilde y, z). \]
For every $\Delta$ such that $y + \Delta \notin S$, we have that $\tilde y + \Delta \notin S$ since $\tilde y \leq y$ and $S$ is upward closed.
The upward closedness of $S$ also implies that since $y \in S$, there must exist some coordinate $k$ such that $\Delta_k < 0$.
From the intensity matrix $Q_C$, there are three forms that $\Delta$ can take such that $Q_C(y, y + \Delta) > 0$: (1) $\Delta_i^{(1)} = e_{A_i} - e_{P_i}$, (2) $\Delta_i^{(2)} = e_{P_i} - e_{A_i} $, or (4) $\Delta_i^{(4)}= -e_{A_i}$, for some $i \in \calG$.
Note that for (2) and (4), $\tilde y \leq y$ and $\tilde y \in S$ immediately imply that $\tilde y_{A_i} = y_{A_i}$; otherwise, if $\tilde y_{A_i} < y_{A_i}$, then $\tilde y \leq y - e_{A_i} \leq y + \Delta$, which, by the upward closedness of $S$, contradicts $y + \Delta \notin S$.

Fix $\Delta$ such that $\tilde y + \Delta \notin S$ and first suppose $\Delta = \Delta_i^{(1)}$ for some $i \in \calG$.
In this case, we have
\[  Q_C\left(y, y + \Delta_i^{(1)}\right) = \lambda_i = \tilde Q_C\left(\tilde y, \tilde y + \Delta_i^{(1)}\right). \]
Next suppose $\Delta = \Delta_i^{(2)}$ for some $i \in \calG$.
Since $y_{A_i} = \tilde y_{A_i}$,
\[ Q_C\left(y, y + \Delta_i^{(2)}\right) = y_{A_i} \cdot \mu_i = \tilde y_{A_i} \cdot \mu_i = \tilde Q_C\left(\tilde y, \tilde y + \Delta_i^{(2)}\right). \]
Lastly, suppose $\Delta = \Delta_i^{(3)}$.
We have
\[ Q_C\left(y, y + \Delta_i^{(3)}\right) = \alpha \cdot \sum_{j \in \calB} \bbP_C \left[ R_{ij} | y \right] \cdot \frac{x_{ij}^*}{w_i} \leq \alpha \cdot \frac{\lambda_i}{w_i} = \tilde Q_C\left(\tilde y, \tilde y + \Delta_i^{(3)}\right), \]
where the inequality follows from the fact that $\bbP_C \left[ R_{ij} | y \right] \leq 1$ trivially and Constraint~\ref{eqn:flow-constraint-seller} from $\mathrm{LP}_\mathrm{off}$.

Therefore,
\[ \sum_{\Delta : y + \Delta \in \calY_C \setminus S} Q_C(y, y + \Delta) \leq \sum_{\Delta : \tilde y + \Delta \in \calY_C \setminus S} \tilde Q_C(\tilde y, \tilde y + \Delta). \qedhere \]
\end{proof}

\begin{comment}
\LowerBoundCompetitiveRatioMultiGood*
\begin{proof}
Rearranging terms, the claim is equivalent to
\[ \left(1 + \sum_{q=1}^2 \prod_{r=1}^q \frac{z}{r + \frac{3}{4} \cdot \frac{z}{1 - \exp(-z)}} \right) \left(3 + 4 \cdot \exp(-z) \right) \geq 7 \]
for all $z \in \mathbb{R}_{\geq 0}$. The left-hand side of the expression above is increasing in $z$,\todo{Necessary to prove increasing or sufficient to just state?} and so, for all $z \in \R_{\geq 0}$,
\begin{align*}
    &\left(1 + \sum_{q=1}^2 \prod_{r=1}^q \frac{z}{r + \frac{3}{4} \cdot \frac{z}{1 - \exp(-z)}} \right) \left(3 + 4 \cdot \exp(-z) \right)\\ \geq \lim_{z \rightarrow 0^+} & \left(1 + \sum_{q=1}^2 \prod_{r=1}^q \frac{z}{r + \frac{3}{4} \cdot \frac{z}{1 - \exp(-z)}} \right) \left(3 + 4 \cdot \exp(-z) \right) = 7. \qedhere
\end{align*}
\end{proof}
\end{comment}

\PrJReachesI*
\begin{proof}
Fix some good $i \in \calG$ and some buyer $j \in \calB$ and consider \Cref{alg:main} when a buyer of type $j$ arrives.
By the upper bound of \Cref{cor:Pr-available-lower-upper-bounds}, the probability that the seller permits the sale of some good $i^\prime \neq i$ that is present and precedes good $i$ in the randomly chosen ordering $\sigma$ is
\[ \sum_{i^\prime \in \calG} \bbP \left [ \sigma(i^\prime) < \sigma(i) \right] \cdot \alpha \cdot \frac{x_{i^\prime j}^*}{\gamma_j \cdot w_{i^\prime}} \cdot \left( 1 - \left(1 + \sum_{q=1}^C \left(\frac{\lambda_{i^\prime}}{\mu_{i^\prime}} \right)^q \right)^{-1} \right) \leq \frac{\alpha}{2} \cdot \sum_{i^\prime \in \calG} \frac{x_{i^\prime j}^*}{\gamma_j} \leq \frac{\alpha}{2}. \]
The first inequality holds since $w_{i^\prime} = 1 - \exp(- \lambda_{i^\prime} / \mu_{i^\prime}) \geq 1 - (1 + \sum_{q=1}^C \frac{\lambda_{i^\prime}}{\mu_{i^\prime}})^{-1}$ for all $C$, and the second inequality follows from Constraint \eqref{eqn:flow-constraint-buyer} from $\mathrm{LP}_\mathrm{off}$.
\end{proof}

%\section{Omitted proofs of \Cref{sec:impossibility}}\label{appendix:impossibility}

%\section{Omitted proofs of \Cref{sec:mechanism-design}}
%\label{appendix:mechanism-design}

\section{Mechanism design implications}\label{appendix:mechanism-design}

By standard results in mechanism design, our posted-price mechanisms imply truthful mechanisms for approximating the optimal social welfare and revenue of any mechanism for this problem. 
For completeness, we discuss implications of our policies to mechanism design for the single-good problem in our model. 

In the mechanism design setting, buyers arrive as before with their real value $v_j$ for the good again drawn from the known distribution $\calD$, though here they may misreport their true value for the good. 
Our randomized posted-price mechanism sets a price-probability pair $(\bar{v},\bar{p})$, and upon arrival of a buyer who observes an available item and bids $b$, the mechanism sells the item at price $\bar{v}$ if the reported bid is strictly higher than $\bar{v}$, and sells the item at the same price with probability $\bar{p}$ if the bid is equal to $\bar{v}$ (else, an item is not sold).
The utility of the buyer with true valuation $v_j$ for the good is $v_j-\bar{v}$. 
It is easy to see that this policy is weakly DSIC (dominant-strategy incentive compatible) and individually rational.\footnote{For weakly DSIC mechanisms, a weakly dominant strategy for a buyer is to report his true value. An individually rational mechanism is one in which a buyer loses nothing from entering the market.}
Since the social welfare (buyer utility + seller revenue) is precisely the gain, $v_j = (v_j - \bar{v}) + \bar{v}$, it is immediate that our posted-price policies yield weakly DSIC and individually rational mechanisms which match our policies' approximation ratios.

\begin{cor}
There exists a weakly DSIC and individually rational mechanism for the single-good stationary prophet inequality problem whose expected social welfare is a $\nicefrac{1}{2}$- and $0.656$-approximation of the social-welfare-maximizing offline and online mechanisms, respectively.
\end{cor}

It is quite standard in the literature to go through Myerson's lemma \cite{myerson1981optimal} to leverage mechanisms which approximately maximize social welfare to obtain mechanisms which approximately maximize the seller's revenue (going via so-called ``virtual valuations'') . Unfortunately, this approach does not work for randomized policies under discrete distributions over buyer valuations, as in our setting. However, it is not hard to leverage our randomized posted-price policies to obtain deterministic posted-price policies, for which the above equivalence of Myerson does hold.

\begin{restatable}{cla}{RandomizedRoundedToDeterministic}
\label{cla:randomized-rounded-to-deterministic}
For any randomized posted-price policy $\calP$ for the single-good stationary prophet inequality problem, there exists a deterministic posted-price mechanism with at least half the expected average gain of $\calP$.
\end{restatable}
\begin{proof}
Fix a randomized posted-price mechanism $\calP$ for the stationary prophet inequality problem characterized by price-probability pair $(\bar v, \bar p)$.
If $\bar p \in \{0, 1\}$, then $\calP$ is in fact a deterministic posted-price mechanism and the claim is trivially true.
Define two deterministic posted-price mechanisms which we refer to as $\calP^0$ and $\calP^1$.
In particular, mechanism $\calP^0$ accepts all bids strictly greater than $\bar v$ but rejects any bid of $\bar v$ or less.
On the other hand, mechanism $\calP^1$ accepts all bids greater than or equal to $\bar v$ but rejects any bid strictly less than $\bar v$.
We denote the expected average gains of $\calP^0$ and $\calP^1$ as $r^{\calP^0}$ and $r^{\calP^1}$, respectively. 

Let $r_{>\bar v}^\calP$ denote the expected average gain under $\calP$ from selling the good to buyers with value strictly greater than $\bar v$ and let $r_{=\bar v}^\calP$ denote the expected average gain under $\calP$ from selling the good to buyers with value exactly $\bar v$.
Note that we trivially have $r_{>\bar v}^\calP + r_{=\bar v}^\calP = r^\calP$.
It is clear that $r^{\calP^0} \geq r_{>\bar v}^\calP$ since the stationary probability under $\calP^0$ that an item is available when a buyer with value strictly greater than $\bar v$ arrives is at least the same probability under $\calP$.
We also have that $r^{\calP^1} \geq r_{=\bar v}^\calP$.
This is true since $r^{\calP^1}$ is at least the expected average gain of the policy that only sells items to buyers with value exactly $\bar v$, which is trivially an upper bound on $r_{=\bar v}^\calP$.
It follows then that $r^{\calP^0} + r^{\calP^1} \geq r_{>\bar v}^\calP + r_{=\bar v}^\calP = r^\calP$ and therefore,
\[ \max \left\{ r^{\calP^0}, r^{\calP^1} \right\} \geq \nicefrac{1}{2} \cdot r^\calP. \qedhere \]
\end{proof}

Combining the above with our algorithms of \Cref{thm:1/2-competitive-single-good-policy,thm:0.656-approximate-single-good-policy} and with Myerson's lemma then yields the following.

\begin{cor}
There exists a DSIC mechanism for the single-good stationary prophet inequality problem whose expected revenue is a $\nicefrac{1}{4}$- and $0.328$-approximation of the revenue-maximizing offline and online DSIC mechanisms, respectively.
\end{cor}

\bibliographystyle{acmsmall}
\bibliography{abbreviations,bibliography,ultimate}

\end{document}